\documentclass[doc, floatsintext, hidelinks, table]{apa7}

\usepackage{amsthm}
\usepackage{amssymb}
\usepackage{mathtools}
\usepackage[natbib=true,style=authoryear]{biblatex}
\newcommand{\supplement}{
        \setcounter{figure}{0}
        \setcounter{section}{0}
        \setcounter{equation}{0}
        \setcounter{page}{1}
        \renewcommand{\thepage}{S\arabic{page}}
        \renewcommand{\thesection}{\Alph{section}}
        \renewcommand{\thefigure}{\thesection.\arabic{figure}}
}
\usepackage{csquotes}
\usepackage{hyperref}
\hypersetup{
    colorlinks=true,
    linkcolor=red,
    filecolor=.,
    citecolor=blue,
    urlcolor=blue
    }
\newtheorem{theorem}{Theorem}
\newtheorem*{remark}{Remark}

\title{Structured Secant Methods to Select Smoothing Parameters for General Smooth Models}
\shorttitle{}

\author{Joshua Krause, Jelmer P. Borst, Jacolien van Rij}
\affiliation{{Bernoulli Institute for Mathematics, Computer Science, and Artificial Intelligence}}

\abstract{
General smooth models replace parameters of a regular likelihood with additive models. The models can include parametric terms, Gaussian random effects, and smooth functions of covariates. The latter are parameterized via a reduced-rank spline basis and regularized via weighted quadratic penalties placed on the basis coefficients. Estimates for these weights (i.e., \emph{smoothing parameters}) can be obtained by optimizing the Laplace-approximate Bayesian marginal likelihood. Existing (second-order) methods require the Hessian of the log-likelihood to solve this optimization problem approximately -- exact optimization requires up to fourth order derivatives -- which can be difficult to derive and expensive to evaluate. To address these problems, we present a quasi-Newton variant of the second-order Extended Fellner-Schall (EFS) optimization method. Our qEFS method relies on structured limited-memory secant approximations to the Hessian of the log-likelihood and is principally first-order. However, the approximation can also be accumulated for a sub-block of the Hessian, with the remaining columns being constrained to match those of the actual Hessian. The exact columns then provide additional structure for the sub-block approximation, which becomes more accurate as a result. We show that the qEFS method converges to the EFS method under certain conditions and continues to provide good estimates beyond these circumstances, which we illustrate in simulation studies. Secondary tasks involving the Hessian (confidence interval coverage \& model selection) require partial approximations to achieve close to nominal performance. We provide Hidden Markov and Tweedie model examples, for which the qEFS method is substantially easier to implement than alternative methods.
}

\leftheader{Joshua Krause, Jelmer P. Borst, Jacolien van Rij}


\authornote{
   \addORCIDlink{Joshua Krause}{0000-0002-3165-2584}
   \addORCIDlink{Jelmer P. Borst}{0000-0002-4493-8223}
   \addORCIDlink{Jacolien van Rij}{0000-0001-7445-5647}

  Correspondence concerning this article should be addressed to Joshua Krause, Bernoulli Institute for Mathematics, Computer Science and Artificial Intelligence, University of Groningen, Nijenborgh 9, 9747 AG Groningen, Netherlands.  E-mail: j.krause@rug.nl}

\begin{document}
\maketitle
\setlength{\parindent}{0pt}
\newpage

\section{Introduction}\label{sec:Introduction}
Generalized additive models \citep[GAMs;][]{hastie_generalized_1986,wood_generalized_2017-1} are a popular type of generalized linear model, in which the mean of independent response variables depends on an additive combination of potentially non-linear functions of covariates. General smooth models extend this principle to any model with a regular likelihood; individual parameters of the latter (or a known function of them) are parameterized via additive models that can themselves include parametric terms, Gaussian random effects, and smooth functions of covariates. GAMs of location, scale, and shape parameters \citep[GAMLSS;][]{rigby_generalized_2005}, Vector Generalized Additive Models \citep[VGAMs;][]{yee_vector_2015}, Markov-swichting GAMs \citep[][]{langrock_markov-switching_2017}, and multivariate additive smooth models are popular special cases \citep[e.g.,][]{wood_smoothing_2016}. More general Hidden Markov models, with transition probabilities and parameters of the emission process depending on smooth functions of covariates \citep[e.g.,][]{michelot_hmmtmb_2025}, and SEIR models of Covid-19 infections in England, with the effect of non-pharmaceutical interventions varying smoothly over time \citep[e.g.,][]{wood_was_2021}, constitute more exotic examples.

Inference of the smooth functions is complicated by the requirement to estimate \emph{smoothing parameters}, determining the appropriate level of smoothness for each of the functions involved. While \citet{wood_smoothing_2016} provided a unified smoothing parameter estimation framework for the class of general smooth models, their approach comes with a steep theoretical cost: To estimate smoothing parameters for generic models, their method requires up to fourth order derivatives of the model's log-likelihood.

Partially motivated by this problem, \citet{wood_generalized_2017} proposed the Extended Fellner-Schall (EFS) update, an extension of the method by \citet{fellner_robust_1986} and \citet{schall_estimation_1991}. To estimate smoothing parameters for general smooth models, the EFS update only requires the Hessian matrix of the log-likelihood with respect to the coefficients involved in the models of the likelihood parameters. Additionally, the EFS update can exploit sparsity of the Hessian and thus often remains efficient for large (multi-level) models. As such, the EFS update goes a long way towards easing the implementation of novel generic smooth models. However, already the expressions of the Hessian elements can become quite complicated for general models, and evaluating them in computer code is often expensive. Similarly, the Hessian matrix might not be sparse or have a sparsity structure that might not be known in advance. While automatic differentiation (AD) could alleviate some of these problems \citep[e.g.,][]{michelot_hmmtmb_2025}, not every AD implementation includes sparsity detection tools and these will not be of any help if the Hessian is not sparse.

In this paper, we present the structured quasi-EFS (qEFS) update to address these limitations. The qEFS update is principally a first-order optimization method, replacing the actual Hessian of the log-likelihood with a structured quasi-Newton secant approximation to the Hessian \citep[e.g.,][]{dennis_convergence_1989,nocedal_numerical_2006}. As a result, the qEFS update only requires the gradient of the log-likelihood with respect to the coefficients involved in the models of the likelihood parameters -- which can also be obtained via AD or approximated numerically via finite differencing. We show that despite this simplification, the qEFS update converges to the full EFS update under certain conditions, thus achieving the performance of a second-order method at a still lower theoretical cost.

The qEFS update can also rely on \emph{limited-memory} quasi-Newton approximations to the Hessian. At the cost of the aforementioned theoretical convergence guarantees, this allows the new update to keep the count of operations low and to be as memory-efficient as required, regardless of whether the actual Hessian is sparse or not. This is particularly useful when working with large (multi-level) models for which alternative methods might be unfeasible. To compensate for the potential loss in accuracy, the (limited-memory) approximations to the Hessian can also be \emph{partial}, i.e., accumulated only for a sub-block of the Hessian, with the remaining columns being constrained to match those of the actual Hessian or a finite difference approximation. The exact columns then effectively provide further structure for the sub-block approximation, which becomes more accurate as a result \citep[e.g.,][]{dennis_convergence_1989}. In consequence, the partial limited-memory approximations tend to remain accurate in practice, resulting in estimates comparable to the full EFS update -- at a lower theoretical \emph{and} computational cost.

Before introducing the qEFS update in detail, we provide the necessary background for general smooth models and the EFS update by \citet{wood_generalized_2017} in section \ref{sec:Background}. Section \ref{sec:Method} then discusses the qEFS update, which has been implemented in the open-source \texttt{mssm} Python toolbox \citep[v $\geq$ 1.2.5;][]{krause_mixed-sparse-smooth-model_2025} to facilitate practical application. In section \ref{sec:Performance} we show under which conditions the qEFS update converges to the EFS update. We also report results from a series of simulation studies to show the robustness of the new method in practice and to illustrate how MSE and secondary task performance (confidence interval coverage \& model selection) depend on the ratio of Hessian columns approximated. Section \ref{sec:Examples} features examples of more exotic models (e.g., Hidden Markov and Tweedie models of location, scale, and shape) that can more easily be estimated with the new update. We also provide an example of a large multi-level model to demonstrate the substantial efficiency gains possible with the new update. In section \ref{sec:Discussion}, we close with a discussion of the new method itself as well as the usefulness of the presented Hessian approximations for other estimation techniques \citep[e.g., integrated nested Laplace approximation and fully Bayesian inference; see][]{rue_approximate_2009,wood_simplified_2019,wood_inference_2020}.

\section{General Smooth Models}\label{sec:Background}
In this section we provide a brief formal introduction to the general smooth model framework presented by \citet{wood_smoothing_2016}. As mentioned in the introduction, a general smooth model is obtained by replacing parameters $\mu,\phi,...\tau$ of a regular log-likelihood $\mathcal{L(\mu,\phi,...\tau)} = log(p(\mathbf{y}|\mu,\phi,...\tau))$ with $g_\mu(\eta_\mu),g_\phi(\eta_\phi),...,g_\tau(\eta_\tau)$ where the $g$ are parameter-specific link functions with known inverse and $\eta_\mu,\eta_\phi,...,\eta_\tau$ are \emph{linear predictors} of the general form given in (\ref{EQ:eta}).

\begin{equation}\label{EQ:eta}
    \eta = \alpha + \beta_1 z_1 +\ ...\ + f_1(x_1) + f_2(x_{2\,1},x_{2\,2}) +\ ...\
\end{equation}

Here $\alpha$ and $\beta_1$ are parametric terms, which can be used to model linear dependencies of $\eta$ on covariates like $z_1$. Note, that some linear predictors might only include an $\alpha$ term so that constant parameters of the likelihood can be estimated as well (e.g., the scale parameter of a GAM). $f_1$ and $f_2$ are smooth functions of covariates $x_1$ and $x_{2\,1},x_{2\,2}$ respectively. The $f_j$ are parameterized as weighted sums of $k_j$ known functions $b$, so that

\begin{equation}\label{eq:smooth_term}
f_j(x_{j\,1},x_{j\,2},...) = \sum^{k_j}_i b_{j\,i}(x_{j\,1},x_{j\,2},...)\boldsymbol{\beta}_{j\,i}.
\end{equation}

Any spline basis can be used for univariate $b$  \citep[e.g., B-spline basis functions;][]{eilers_flexible_1996}, while $b$ functions of multiple covariates can conveniently be obtained by constructing a tensor smooth basis from univariate basis functions \citep[e.g.,][]{wood_lowrank_2006,wood_straightforward_2013}. We refer to \citet{wood_generalized_2017-1} for an overview of the different options. 

In what follows, we use $\boldsymbol{\beta}$ to refer to the total coefficient vector of dimension $N_p$, including the parametric coefficients and basis weights from all linear predictors. Associated with each of the $f_j$ is a quadratic smoothing penalty $\boldsymbol{\beta}^\top\mathbf{S}^j_{\lambda}\boldsymbol{\beta} =\boldsymbol{\beta}_j^\top\mathcal{S}^j_{\lambda}\boldsymbol{\beta}_j$, capturing the function's complexity \citep[or ``roughness'', see][]{silverman_aspects_1985,wood_smoothing_2016}. Here, $\mathbf{S}^j_{\lambda}$ is a zero-padded version of $\mathcal{S}^j_{\lambda}$, the smooth's weighted penalty matrix \citep[e.g.,][]{wood_smoothing_2016,wood_generalized_2017-1}. Each penalty matrix depends on a subset of smoothing parameters $\boldsymbol{\lambda}_j\subset\boldsymbol{\lambda}$. For univariate smooth terms, $\mathcal{S}^j_{\lambda}$ is often simply set to $\lambda_j\mathcal{S}^j$, where $\mathcal{S}^j$ is a single $k_j*k_j$ penalty matrix of fixed structure. However, smooth terms can also be associated with more complex penalties \citep[e.g., tensor smooths;][]{wood_lowrank_2006}, so that in general $\mathcal{S}^j_{\lambda} = \sum_{l} \lambda_{l\,j}\mathcal{S}^{l\,j}$. Additionally, the same subset of smoothing parameters might be used to penalize different smooth terms \citep[see][for a more detailed discussion]{wood_generalized_2017-1}. To avoid notational clutter we will simply use $\mathcal{S}^r$ ($\mathbf{S}^r$) to refer to the (embedded) penalty matrix associated with an individual parameter $\lambda_r \in \boldsymbol{\lambda}$.

From a Bayesian perspective, the smoothing penalties introduce an (improper) multivariate Gaussian ``smoothness'' prior $\boldsymbol{\beta} \sim N(\mathbf{0},\mathbf{S}_\lambda^-)$ for $\boldsymbol{\beta}$ where $\mathbf{S}_\lambda = \sum^j\mathbf{S}^j_\lambda$ and $\mathbf{S}_\lambda^-$ denotes a pseudo-inverse \citep[e.g.,][]{kimeldorf_correspondence_1970,silverman_aspects_1985,rue_approximate_2009,wood_smoothing_2016}. Note, that the prior is improper because the Kernel of $\mathbf{S}_{\lambda}$ will usually be non-trivial \citep[e.g.,][]{wood_smoothing_2016,wood_generalized_2017-1}. First, $\mathbf{S}_{\lambda}$ will usually contain multiple zero blocks on the diagonal reflecting parametric terms. Similarly, it will often be desirable to avoid penalizing the case of any $f_j$ taking the form of a sufficiently simple (e.g., linear or quadratic) function. Formally, we typically have $\tilde{\boldsymbol{\beta}}_j^\top\mathcal{S}^j_{\lambda}\tilde{\boldsymbol{\beta}}_j=0$ for any constellation $\tilde{\boldsymbol{\beta}}_j$ that would result in a sufficiently simple function $f_j$. The connection between the smoothing penalties and Bayesian prior distributions also implies that i.i.d Gaussian random terms can simply be treated just like any other smooth function \citep[e.g.,][]{kimeldorf_correspondence_1970,wood_generalized_2017-1}.

\subsection{Estimating General Smooth Models}\label{sec:reml}

In essence, the approach taken by \citet{wood_smoothing_2016} is a generalization of the empirical Bayes approach to smoothing parameter selection for Gaussian models proposed by \citet{wahba_comparison_1985}. The idea is to select the smoothing parameters that optimize (approximately) the Bayesian marginal likelihood $p(\mathbf{y}|\boldsymbol{\lambda})$, also referred to as the restricted marginal likelihood \citep[REML;][]{wahba_comparison_1985,wood_fast_2011,wood_smoothing_2016}. Given the density $p(\boldsymbol{\beta}|\boldsymbol{\lambda})$ of the aforementioned improper prior on $\boldsymbol{\beta}$, the marginal likelihood for a general smooth model can be computed as shown in (\ref{EQ:reml}).

\begin{equation}\label{EQ:reml}
\begin{split}
p(\mathbf{y}|\boldsymbol{\lambda}) & = \int p(\mathbf{y}|\boldsymbol{\beta}) p(\boldsymbol{\beta}|\boldsymbol{\lambda})d\boldsymbol{\beta} \\ &= \int p(\mathbf{y},\boldsymbol{\beta}|\boldsymbol{\lambda})d\boldsymbol{\beta} \\ &\approx p(\mathbf{y},\hat{\boldsymbol{\beta}}|\boldsymbol{\lambda})\frac{\sqrt{2\pi}^{N_p}}{|\mathcal{H}|^{0.5}}
\end{split}
\end{equation}

The final result in (\ref{EQ:reml}), obtained by substituting a Laplace approximation for the conditional posterior $\boldsymbol{\beta}|\mathbf{y},\boldsymbol{\lambda}$ \citep[see][for a more detailed discussion]{wood_inference_2020}, yields the Laplace-approximate marginal likelihood $p_L(\mathbf{y}|\boldsymbol{\lambda})$ \citep[exactly the marginal likelihood in the Gaussian case;][]{wahba_comparison_1985,wood_smoothing_2016}. Computing $p_L(\mathbf{y}|\boldsymbol{\lambda})$ requires the maximizer $\hat{\boldsymbol{\beta}}$ of the joint log-likelihood $log(p(\mathbf{y},\boldsymbol{\beta}|\boldsymbol{\lambda}))$ and the negative Hessian $\mathcal{H}$ of the latter, evaluated at $\hat{\boldsymbol{\beta}}$. Considering the density of the prior, $\hat{\boldsymbol{\beta}}$ is readily obtained by optimizing the penalized log-likelihood $\mathcal{L}_\lambda(\boldsymbol{\beta}) = \mathcal{L}(\boldsymbol{\beta}) - \frac{1}{2}\boldsymbol{\beta}^\top\mathbf{S}_\lambda\boldsymbol{\beta}$, while $\mathcal{H} = \mathbf{H} + \mathbf{S}_\lambda$ with $\mathbf{H}=-\frac{\partial^2 \mathcal{L}}{\partial \boldsymbol{\beta} \partial \boldsymbol{\beta}^\top}\Big\rvert_{\hat{\boldsymbol{\beta}}}$ corresponding to the negative Hessian of the log-likelihood evaluated at $\hat{\boldsymbol{\beta}}$ \citep[e.g.,][]{wood_smoothing_2016}.

\citet{wood_smoothing_2016} show how to optimize $\mathcal{V}(\boldsymbol{\lambda})=log(p_L(\mathbf{y}|\boldsymbol{\lambda}))$ exactly, via Newton's method. Apart from $\mathbf{H}$, this approach generally requires third and fourth order derivatives of the log-likelihood. In contrast, the EFS update by \citet{wood_generalized_2017} only requires $\mathbf{H}$ to approximately optimize $\mathcal{V}(\boldsymbol{\lambda})$ (see Appendix \ref{app:efs} of the supplementary materials for a review outlining how the update can be derived). The update, defined in (\ref{eq:efs}) needs to be alternated with optimization of $\mathcal{L}_\lambda$ to obtain $\hat{\boldsymbol{\beta}}$ given the current estimate of $\boldsymbol{\lambda}$.

\begin{equation}\label{eq:efs}
    \lambda_r^* = \lambda_{r} \frac{tr\left(\mathbf{S}^{-}_\lambda \mathbf{S}^r\right) - tr\left(\mathcal{H}^{-1} \mathbf{S}^r\right)}{\boldsymbol{\hat{\beta}}^\top\mathbf{S}^r\boldsymbol{\hat{\beta}}}
\end{equation}

Notably, the update in (\ref{eq:efs}) is only well-defined whenever $\mathbf{H}$ is at least positive semi-definite (PSD) and $\mathbf{H} + \mathbf{S}_\lambda$ is positive definite \citep[PD; see Theorem \& Remark 1 given by][]{wood_generalized_2017}. If $\mathbf{H}$ is indefinite, it should be replaced with the expected negative Hessian $\mathbb{E}\mathbf{H}$ or a PSD approximation to the latter. Additionally, the update is only guaranteed to optimize $\mathcal{V}$ exactly (with step-length control) in case $\mathbf{H}$ does not depend on $\boldsymbol{\lambda}$ (see Appendix \ref{app:efs} of the supplementary materials). This will generally be the case if the log-likelihood is exactly quadratic, such as for a strictly additive model \citep[e.g.,][]{wood_generalized_2017}. Otherwise, as illustrated in Figure \ref{fig:taylor_reml}, the update in (\ref{eq:efs}) essentially takes a step towards the maximum of a ``Taylor-approximate'' version of $\mathcal{V}(\boldsymbol{\lambda})$ obtained by replacing $\mathcal{L}$ with a second-order Taylor approximation around the current estimate $\hat{\boldsymbol{\beta}}$. In consequence, the criterion optimized by the EFS update changes with $\boldsymbol{\lambda}$. As pointed out by \citet{wood_generalized_2017}, this is similar to the PQL/POI approach to GAM estimation \citep[e.g.,][]{breslow_approximate_1993,gu_cross-validating_1992}, which also neglects changes in the Hessian with $\boldsymbol{\lambda}$ \citep[see also][for a discussion]{wood_generalized_2017-2}.

\begin{figure*}[!h]
    \caption{Illustration of the EFS Update for General Smooth Models}
    \includegraphics[width=\textwidth]{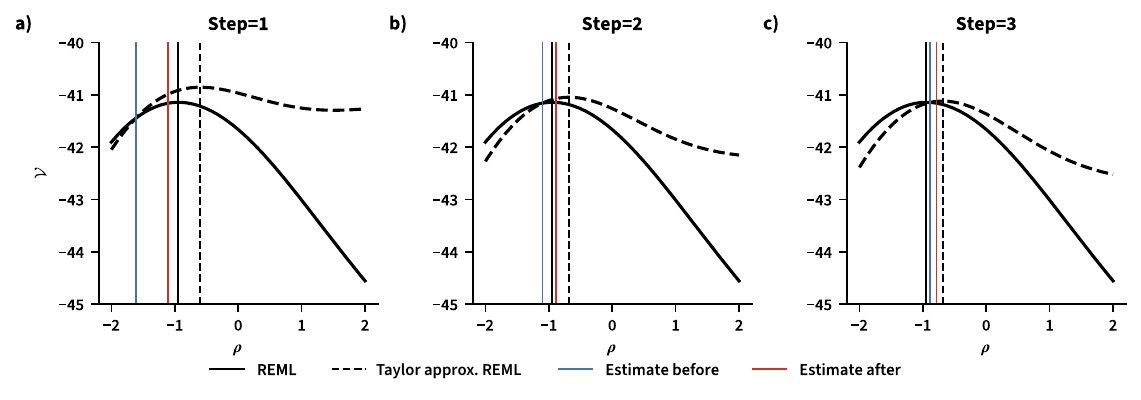}
    \label{fig:taylor_reml}
    {\small
         Panels a - c illustrate the first three EFS updates applied to a single $\lambda$ parameter for a proportional hazard model including a single smooth term. The solid black curve and vertical line correspond to $\mathcal{V}$ and its maximizer. The blue vertical line corresponds to the current estimate of $\lambda$ (i.e., before the EFS update). The dashed curve and vertical line correspond to the ``Taylor-approximate'' version of  $\mathcal{V}$ and its maximizer respectively. The EFS update takes steps (red vertical line) towards the maximum of the approximate criterion, which is different at every step because the Taylor approximation changes after every update to $\lambda$.
        }
\end{figure*}

 The fact that the approximate criterion changes with $\boldsymbol{\lambda}$ also complicates step-length control. One option is to check whether the EFS update still corresponds to an ascent direction for the next ``Taylor-approximate`` criterion \citep[cf.][]{wood_generalized_2017-1}. Alternatively, step-length can simply be omitted since, as pointed out by \citet{wood_generalized_2017}, it appears to be rarely needed in practice to ensure convergence.

 This concludes our review of general smooth models and the EFS update. After a brief introduction to quasi-Newton approaches, the next section presents the structured quasi-Newton variant of the EFS update (qEFS), in which $\mathbf{H}$ is replaced with a structured limited-memory quasi-Newton secant approximation to the negative Hessian.

\section{The qEFS Update to Smoothing Parameters}\label{sec:Method}
In the context of smooth models, quasi-Newton approaches like the popular BFGS update \citep[e.g.,][]{nocedal_numerical_2006} have previously been used to find the estimate $\hat{\boldsymbol{\beta}}$ which maximizes $\mathcal{L}_\lambda$ \citep[e.g.,][]{pya_shape_2015}. These methods substitute an approximation matrix $\hat{\mathcal{H}}^{i}$ for the exact negative Hessian $\mathcal{H}$ to compute the quasi-Newton direction $\mathbf{d}_i=(\hat{\mathcal{H}}^{i})^{-1}\mathbf{g}_i$ with $\mathbf{g}_i=\frac{\partial \mathcal{L_\lambda}}{\partial \boldsymbol{\beta}}\Big\rvert_{\hat{\boldsymbol{\beta}}_{i-1}}$ and $\hat{\boldsymbol{\beta}}_{i-1}$ denoting the estimate of the coefficients at the start of iteration $i$ of the quasi-Newton routine. Quasi-Newton updates are often paired with a line search, which scales the direction $\mathbf{d}_i$ by a factor $\alpha$, to ensure that sufficient progress is made on the optimization problem \citep[see][for a discussion of the Wolfe or Armijo conditions typically imposed on $\alpha$]{nocedal_numerical_2006}. Combined, this results in $\hat{\boldsymbol{\beta}}_{i} = \hat{\boldsymbol{\beta}}_{i-1} + \mathbf{s}_i$ with (scaled) step $\mathbf{s}_i=\alpha\mathbf{d}_i$.

Aside from $\hat{\boldsymbol{\beta}}_{i-1}$, the current approximation matrix $\hat{\mathcal{H}}^{i-1}$ is also updated. Independent of the specific method, the update $\hat{\mathcal{H}}^{i}$ is typically chosen from the set of \emph{symmetric solutions} to the \emph{secant equation}\footnote{Focus on the secant equation is motivated by the fact that the equation holds approximately for a general log-likelihood when substituting the exact negative Hessian, with the error resulting from an inexact second order Taylor approximation.} $\hat{\mathcal{H}}^{i}\mathbf{s}_i = \mathbf{v}_i$, where $\mathbf{v}_i= -\frac{\partial \mathcal{L}_\lambda}{\partial \boldsymbol{\beta}}\Big\rvert_{\hat{\boldsymbol{\beta}}_i} + \frac{\partial \mathcal{L}_\lambda}{\partial \boldsymbol{\beta}}\Big\rvert_{\hat{\boldsymbol{\beta}}_{i-1}}$ denotes the difference in the negative gradient of the penalized log-likelihood after taking the step $\mathbf{s}_i$ \citep[e.g.,][]{nocedal_numerical_2006}. Different quasi-Newton methods then impose different additional constraints on the Hessian update. For example, the BFGS update ensures that the approximation matrix remains positive definite \citep[e.g.,][]{nocedal_numerical_2006}.

Considering that $\hat{\mathcal{H}}$ is used as a replacement for $\mathcal{H}$ to compute the Newton step, it might be tempting to simply substitute the final approximation matrix for the exact Hessian in (\ref{eq:efs}). Unfortunately, this is unlikely to result in reasonable updates to $\boldsymbol{\lambda}$. For once, there is no guarantee that the approximation decomposes into a sum as neatly as the exact Hessian $\mathcal{H} =  \mathbf{H} + \mathbf{S}_\lambda$. In fact, $\hat{\mathcal{H}} - \mathbf{S}_\lambda$ can readily become indefinite, in which case the update in (\ref{eq:efs}) would be ill-defined \citep[][]{wood_generalized_2017}.

We propose to address these problems by performing a structured (or partial) quasi-Newton update \citep[e.g.,][]{dennis_convergence_1989,liu_partial-quasi-newton_2022}. Specifically, rather than maintaining a quasi-Newton approximation for $\mathcal{H}$, we propose to maintain an approximation for $\mathbf{H}$, which is chosen to satisfy the \emph{structured} secant equation in (\ref{eq:structured_secant}).

\begin{equation}\label{eq:structured_secant}
\begin{split}
    &(\hat{\mathbf{H}}_i+\mathbf{S}_\lambda)\mathbf{s}_i = \mathbf{v}_i\\
    &\hat{\mathbf{H}}_i\mathbf{s}_i = \mathbf{v}_i'\ \text{with}\ \mathbf{v}_i' = \mathbf{v}_i - \mathbf{S}_\lambda\mathbf{s}_i
\end{split}
\end{equation}

Note, that 

\begin{equation}\label{eq:v_sharp}
    \mathbf{v}_i' = -\frac{\partial \mathcal{L}}{\partial \boldsymbol{\beta}}\Big\rvert_{\hat{\boldsymbol{\beta}}_i} + \frac{\partial \mathcal{L}}{\partial \boldsymbol{\beta}}\Big\rvert_{\hat{\boldsymbol{\beta}}_{i-1}},
\end{equation}

that is, $\hat{\mathbf{H}}$ is chosen to satisfy the secant equation for the un-penalized log-likelihood $\mathcal{L}$ and so that $\hat{\mathbf{H}}+\mathbf{S}_\lambda$ satisfies the secant equation for the penalized log-likelihood $\mathcal{L}_\lambda$. By definition $\hat{\mathbf{H}}$ is a far more reasonable approximation to the negative Hessian of the log-likelihood compared to $\hat{\mathcal{H}} - \mathbf{S}_\lambda$. However, $\mathbf{H}$ might not even be PSD at convergence of the quasi-Newton routine, in which case trying to maintain a PD approximation to $\mathbf{H}$, for example via BFGS updating, can result in poor approximations to the actual Hessian \citep[e.g.,][]{conn_convergence_1991,nocedal_numerical_2006}.

Therefore, we propose to rely on Symmetric Rank 1 (SR1) updating of $\hat{\mathbf{H}}^{i-1}$. SR1 updating does not enforce positive-definiteness of the approximation, which ensures that $\hat{\mathbf{H}}$ can become PSD or even indefinite \citep[e.g.,][]{conn_convergence_1991,byrd_representations_1994}. We show how to efficiently manipulate the spectrum of $\hat{\mathbf{H}}$ to obtain an additional PSD (or PD) approximation $\hat{\mathbf{H}}_+$ (section \ref{sec:structuredQN}) that can be substituted into (\ref{eq:efs}) to obtain the desired qEFS update. Subsequently, we generalize the structured method so that it can be used to approximate only a sub-block of $\mathbf{H}$, with the remaining columns (and rows) being constrained to match those of a (PSD or PD approximation to) the actual Hessian or a finite difference approximation (section \ref{sec:structuredQN2}).

\subsection{Structured SR1 Secant Approximations of the Hessian}\label{sec:structuredQN}

We start by showing how to perform SR1 updating of $\hat{\mathbf{H}}^{i-1}$ \citep[see also][]{byrd_representations_1994}. Then we show how to manipulate the spectrum of $\hat{\mathbf{H}}$ to guarantee it is PSD (or PD), which can be achieved efficiently when relying on a \emph{limited-memory} SR1 approximation to $\mathbf{H}$. The SR1 update to $\hat{\mathbf{H}}^{i-1}$ satisfying (\ref{eq:structured_secant}) is defined as in (\ref{eq:SR1U}).

\begin{equation}\label{eq:SR1U}
    \hat{\mathbf{H}}^{i} = \hat{\mathbf{H}}^{i-1} + \frac{\left(\mathbf{v}_i'-\hat{\mathbf{H}}^{i-1}\mathbf{s}_i\right)\left(\mathbf{v}_i'-\hat{\mathbf{H}}^{i-1}\mathbf{s}_i\right)^\top\mathbf{s}_i}{\left(\mathbf{v}_i'-\hat{\mathbf{H}}^{i-1}\mathbf{s}_i\right)^\top\mathbf{s}_i}
\end{equation}

The update in (\ref{eq:SR1U}) is ill-defined if $\left(\mathbf{v}_i'-\hat{\mathbf{H}}^{i-1}\mathbf{s}_i\right)^\top\mathbf{s}_i = 0$, in which case it is typically skipped \citep[][]{nocedal_numerical_2006}.

A \emph{limited-memory} (SR1) approximation to $\mathbf{H}$ is obtained by maintaining a queue of only the last $N_u$ pairs of update vectors \citep[e.g.,][]{nocedal_updating_1980,byrd_representations_1994,nocedal_numerical_2006}. Specifically, at every iteration a new pair ($\mathbf{v}_i$, $\mathbf{s}_i$) enters the queue and, if the queue is full, the oldest pair is discarded. The update in (\ref{eq:SR1}) is then applied sequentially to the pairs currently in the queue to obtain $\hat{\mathbf{H}}^{i}$, starting from an initial approximation $\mathbf{H}_0^i = \gamma_i\mathbf{I}$ where $\mathbf{I}$ is an identity matrix and $\gamma_i > 0$ is a scaling factor that can vary with iterations $i$ \citep[e.g.,][]{byrd_representations_1994}. In practice, this update sequence does not have to be evaluated explicitly. Instead, assuming that the update in (\ref{eq:SR1U}) is defined for the $min(i,N_u)$ update vectors in the queue, a compact (limited-memory) representation of $\hat{\mathbf{H}}^{i}$ can be computed directly as shown in \citep[\ref{eq:SR1}; see also][]{byrd_representations_1994,brust_solving_2017}.

\begin{equation}\label{eq:SR1}
\begin{split}
    & \hat{\mathbf{H}}^{i} = \mathbf{H}_0^i + \mathbf{U}^i\mathbf{B}^i\mathbf{U}^{i\top}\ \text{with}\\
    & \mathbf{U}^i = \mathrm{V}^i - \mathbf{H}_0\mathrm{S}^i\ \text{and}\\
    & \mathbf{B} = \left(\mathbf{C}^i + \mathbf{L}^i + \mathbf{L}^{i\top} -\mathrm{S}^{i\top}\mathbf{H}^0\mathrm{S}^{i}\right)^{-1}
\end{split}
\end{equation}

The columns of the $N_p*min(i,N_u)$ matrices $\mathrm{V}^i = [\mathbf{v}_{i-N_u+1}',\mathbf{v}_{i-N_u+2}',...,\mathbf{v}_{i}']$ and $\mathrm{S}^i = [\mathbf{s}_{i-N_u+1},\mathbf{s}_{i-N_u+2},...,\mathbf{s}_{i}]$ hold the last $N_u$ (structured) update vectors while $\mathbf{C}^i$ is diagonal with elements $\mathbf{C}^i_{l\,l}=\mathbf{s}_l^\top\mathbf{v}_l'$ and $\mathbf{L}^i$ is defined as in \citep[\ref{eq:SR1L}; see][]{byrd_representations_1994}.

\begin{equation}\label{eq:SR1L}
  \mathbf{L}^i_{l\,j} =
  \begin{cases}
    \mathbf{s}_l^\top\mathbf{v}_j' & \text{if $l > j$}\\
    0 & \text{otherwise}
  \end{cases}
\end{equation}

In practice it is necessary to actually check for the assumption that (\ref{eq:SR1U}) is well-defined for all update vector pairs in the queue at iteration $i$. Such a check can be based on a triangular decomposition of $\mathbf{B}^i$ \citep[e.g.,][]{byrd_representations_1994} and the representation in (\ref{eq:SR1}) can then simply be computed with any problematic vector pairs excluded. Note, that the ``implicit'' or ``compact'' representation in (\ref{eq:SR1}) can be used directly to perform many operations (e.g., matrix products) involving matrix $\hat{\mathbf{H}}^i$ without forming the latter explicitly \citep[see][]{nocedal_updating_1980,byrd_representations_1994,nocedal_numerical_2006,brust_solving_2017}. Since $\mathbf{H}_0^i$ is very sparse and $\mathbf{U}^i$ is a thin matrix with $N_u << N_p$ columns, this can be far more efficient than what could be achieved with an explicit representation of $\hat{\mathbf{H}}^i$. 

At any iteration $i$ of the quasi-Newton routine, the nearest\footnote{In terms of the Frobenius norm \citep[e.g.,][]{higham_computing_1988}.} PSD (or PD) approximation $\hat{\mathbf{H}}^{i}_+$ to $\hat{\mathbf{H}}^{i}$ can then be obtained by relying on an implicit Eigen-decomposition of $\mathbf{H}_0^i + \mathbf{U}^i\mathbf{B}^i\mathbf{U}^{i\top}$. Appendix (\ref{app:ImplicitEigen}) of the supplementary materials includes a review of the implicit Eigen-decomposition approach, proposed originally by \citet{burdakov_efficiently_2017}, with the main result stated here in (\ref{eq:PSDSR1}).

\begin{equation}\label{eq:PSDSR1}
\begin{split}
    & \hat{\mathbf{H}}^{i}_+ = \mathbf{H}_0^i + \mathbf{P}^i\boldsymbol{\Sigma}_+^i\mathbf{P}^{i\top}\ \text{with}\ \mathbf{P}^i=\mathbf{Q}^i\mathrm{U}^i\text{,}\\
    & \boldsymbol{\Sigma}^i_{+\,l\,l} = \boldsymbol{\Sigma}^i_{l\,l} - min(0,\boldsymbol{\Sigma}^i_{l\,l} + \gamma_i)\text{,}\\
    & \mathbf{R}^i\mathbf{B}^i\mathbf{R}^{i\top} = \mathrm{U}^i\boldsymbol{\Sigma}^i\mathrm{U}^{i\top}\ \text{, and}\ \mathbf{Q}^i\mathbf{R}^i=\mathbf{U}^i
\end{split}
\end{equation}

$\mathbf{U}^i$, $\mathbf{B}^i$, and $\mathbf{H}_0^i$ are defined as in (\ref{eq:SR1}). $\mathbf{Q}^i\mathbf{R}^i=\mathbf{U}^i$ is obtained by forming a ``thin'' QR-decomposition so that $\mathbf{Q}^i$ has $N_u$ columns and $\mathbf{R}^i$ is a $N_u*N_u$ upper-triangular matrix \citep[e.g.,][]{golub_matrix_2013}. As a result, the Eigendecomposition $\mathbf{R}^i\mathbf{B}^i\mathbf{R}^{i\top} = \mathrm{U}^i\boldsymbol{\Sigma}^i\mathrm{U}^{i\top}$ only has to be formed for a $N_u*N_u$ matrix. By definition of $\boldsymbol{\Sigma}^i_{+}$ in line 2 of (\ref{eq:PSDSR1}), it is clear that addition of $\mathbf{H}_0^i=\mathbf{I}\gamma_i$ to $\mathbf{P}^i\boldsymbol{\Sigma}_+^i\mathbf{P}^{i\top}$ results in a positive semi-definite matrix, since the Eigenvalues from $\boldsymbol{\Sigma}^i_{+}$ have been shifted up by just enough, to ensure that negative ones become zero when $\gamma_i$ is added. To ensure that $\hat{\mathbf{H}}^{i}_+$ is numerically PD rather than PSD, a small positive value can be added to all Eigenvalues $l$ for which $\boldsymbol{\Sigma}^i_{+\,l\,l} + \gamma_i = 0$.

The next quasi-Newton direction $\mathbf{d}_i$ can now be computed as $\mathbf{d}_i = (\hat{\mathbf{H}}^{i-1}_+ + \mathbf{S}_\lambda)^{-1}\mathbf{g}_i$. Similarly, the final approximation $(\hat{\mathbf{H}}_++\mathbf{S}_\lambda)^{-1}$ at convergence of the quasi-Newton routine can be substituted for $\mathcal{H}^{-1}$ in the EFS update defined in (\ref{eq:efs}). The required inverse can be obtained using the modified Woodburry identity from \citet{henderson_deriving_1981} which, as shown in (\ref{eq:WBinv}), has the advantage that the inverse can be represented implicitly as well.

\begin{equation}\label{eq:WBinv}
\begin{split}
    & (\hat{\mathbf{H}}_+ + \mathbf{S}_\lambda)^{-1} = (\mathbf{H}_0  + \mathbf{S}_\lambda)^{-1} - \mathbf{M}\mathbf{A}^{-1}\mathbf{N}\ \text{with}\\
    & \mathbf{M} = (\mathbf{H}_0  + \mathbf{S}_\lambda)^{-1}\mathbf{P}\text{,}\
    \mathbf{N} = \boldsymbol{\Sigma}_+\mathbf{P}^\top(\mathbf{H}_0  + \mathbf{S}_\lambda)^{-1}\text{, and}\\
    &\mathbf{A}=\mathbf{I} + \boldsymbol{\Sigma}_+\mathbf{P}^\top(\mathbf{H}_0  + \mathbf{S}_\lambda)^{-1}\mathbf{P}
\end{split}
\end{equation}

Here $\mathbf{H}_0$, $\mathbf{P}$, and $\boldsymbol{\Sigma}_+$ are defined as in (\ref{eq:SR1}) and (\ref{eq:PSDSR1}) and super-scripts indicating the iteration have been omitted to avoid clutter. $\mathbf{A}$, while not symmetric, is non-singular, so the inverse $\mathbf{A}^{-1}$ is guaranteed to be defined at every iteration $i$ of the quasi-Newton routine \citep[see][]{henderson_deriving_1981}. While the inverse $(\mathbf{H}_0  + \mathbf{S}_\lambda)^{-1}$ is theoretically of a $N_p *N_p$ matrix, the latter continues to be symmetric and extremely sparse even for large multi-level models and can thus be computed efficiently by means of a Cholesky decomposition that pivots for sparsity \citep[e.g.,][]{wood_generalized_2017}. Finally, note that the identity in (\ref{eq:WBinv}) holds as long as $(\hat{\mathbf{H}}_+ + \mathbf{S}_\lambda)^{-1}$ exists, even if $\hat{\mathbf{H}}_+$ is only PSD or if $\hat{\mathbf{H}}_+=\mathbf{H}_0 + \mathbf{P}\boldsymbol{\Sigma}_+\mathbf{K}$ where $\mathbf{K}\neq\mathbf{P}^{\top}$ \citep[see][]{henderson_deriving_1981}. Still, at least for the quasi-Newton update, $\hat{\mathbf{H}}^{i-1}_+$ should be set to (numerically) PD to guarantee that $\mathbf{d}_i$ is an ascent direction.

Maintaining an implicit representation of $(\hat{\mathbf{H}}_++\mathbf{S}_\lambda)^{-1}$ has the advantage, that the trace term in (\ref{eq:efs}) can also be computed implicitly, as shown in (\ref{eq:lbfgsTrace}).

\begin{equation}\label{eq:lbfgsTrace}
tr((\hat{\mathbf{H}}_+ + \mathbf{S}_\lambda)^{-1}\mathbf{S}^r)=\sum_l^{N_p} \left(\mathbf{H}_0  + \mathbf{S}_\lambda\right)^{-1}_l\mathbf{S}_l^{r\,\top} - \mathbf{M}_l\mathbf{A}^{-1}(\mathbf{N}\mathbf{S}_l^{r\,\top})
\end{equation}

Here $\left(\mathbf{H}_0  + \mathbf{S}_\lambda\right)^{-1}_l$ denotes row $l$ of matrix $\left(\mathbf{H}_0  + \mathbf{S}_\lambda\right)^{-1}$, $\mathbf{S}_l^{r\,\top}$ denotes column $l$ of matrix $\mathbf{S}^r$, and $\mathbf{M}_l$ denotes row $l$ of matrix $\mathbf{M}$. Note that the sum will rarely have to be computed over all $N_p$ elements, since matrices like $\mathbf{S}^r$ will only contain a handful of non-zero columns, the indices of which will generally be known in advance. Additionally, the matrix vector product $\mathbf{N}\mathbf{S}_l^{r\,\top}$ will fully benefit from the small number of non-zero rows in non-zero columns $\mathbf{S}_l^{r\,\top}$ of $\mathbf{S}^r$, so that each trace -- and ultimately the (quasi-variant of the) update in (\ref{eq:efs}) -- can be computed very efficiently.

\subsection{Incorporating Additional Knowledge About the Hessian}\label{sec:structuredQN2}

Often, it will be possible to evaluate $\mathbf{H}$ or $\mathbb{E}\mathbf{H}$ at least partially, either because certain columns have simple analytical solutions, because individual columns can be approximated via AD, or because the available computational budget permits obtaining finite difference approximations of individual columns. In this section, we show how to incorporate this additional information into the PSD (or PD) approximations $\hat{\mathbf{H}}_+$ \citep[see][for a similar approach to minimax optimization]{liu_partial-quasi-newton_2022}.

We start by assuming (w.l.o.g.), that $\boldsymbol{\beta} = [\mathcal{B}^\top,\mathbf{b}^\top]^\top$ where $\mathcal{B}$ denotes the set of $N_\mathcal{B}$ coefficients for which we can easily evaluate the corresponding column of the Hessian. We can define a corresponding block-factorization for the negative Hessian as shown in (\ref{eq:blockH}).

\begin{equation}\label{eq:blockH}
    \mathbf{H}^i =
    \begin{bmatrix}
        \mathbf{H}^i_{\mathcal{B}\mathcal{B}} & \mathbf{H}^i_{\mathcal{B}\mathbf{b}}\\
        \mathbf{H}^i_{\mathbf{b}\mathcal{B}} & \mathbf{H}^i_{\mathbf{b}\mathbf{b}} \\
    \end{bmatrix}
\end{equation}

Here, $\mathbf{H}_{\mathcal{B}\mathcal{B}}$ for example contains $\frac{\partial^2 \mathcal{L}}{\partial \mathcal{B}\partial\mathcal{B}^\top}\Big\rvert_{\hat{\boldsymbol{\beta}}_{i}}$. We now seek an approximation to (\ref{eq:blockH}) as outlined in (\ref{eq:psdblockH}).

\begin{equation}\label{eq:psdblockH}
    \hat{\mathbf{H}}^i =
    \begin{bmatrix}
        \mathbf{H}^i_{\mathcal{B}\mathcal{B}} & \mathbf{H}^i_{\mathcal{B}\mathbf{b}}\\
        \mathbf{H}^i_{\mathbf{b}\mathcal{B}} & \hat{\mathbf{H}}^i_{\mathbf{b}\mathbf{b}} \\
    \end{bmatrix}
\end{equation}

We are constrained by the fact that we need to be able to transform $\hat{\mathbf{H}}^i$ into a PSD (or PD) approximation for (\ref{eq:efs}) to be defined. Additionally, $\hat{\mathbf{H}}^i_{\mathbf{b}\mathbf{b}}$ should be obtained via quasi-Newton updating, using $min(i,N_u)$ available update vector pairs $(\mathbf{v}_i,\mathbf{s}_i)$. To develop such an approximation we start by assuming that $\mathbf{H}^i_{\mathcal{B}\mathcal{B}}$ is PD so that the inverse $\left(\mathbf{H}^i_{\mathcal{B}\mathcal{B}}\right)^{-1}$ is well defined. As illustrated in section \ref{sec:idpsdapprox}, the PD assumption is technically not necessary but simplifies derivation of the update. For now we assume that, whenever $\mathbf{H}^i_{\mathcal{B}\mathcal{B}}$ is not of full rank in practice, we replace it with a PD approximation $\mathbf{H}^i_{+\,\mathcal{B}\mathcal{B}}$, either obtained through Eigendecomposition or by perturbing the matrix with a scaled identity until a Cholesky succeeds if the Eigendecomposition is too expensive \citep[e.g.,][]{wood_smoothing_2016}. 

Next, we consider the Schur complement defined in Equation (\ref{eq:Schur}), which we will refer to as $\mathbf{D}^i$.

\begin{equation}\label{eq:Schur}
     \mathbf{D}^i \coloneq \mathbf{H}^i/\mathbf{H}^i_{\mathbf{b}\mathbf{b}} = \mathbf{H}^i_{\mathbf{b}\mathbf{b}} - \mathbf{H}^i_{\boldsymbol{b}\mathcal{B}}(\mathbf{H}^i_{\mathcal{B}\mathcal{B}})^{-1}\mathbf{H}^i_{\mathcal{B}\boldsymbol{b}}
\end{equation}

Clearly, if a (PSD) quasi-Newton approximation $\hat{\mathbf{D}}^i$ to the Schur complement could be obtained, this would provide the required approximation $\hat{\mathbf{H}}^i_{\mathbf{b}\mathbf{b}}= \hat{\mathbf{D}}^i + \mathbf{H}^i_{\mathbf{b}\mathcal{B}}(\mathbf{H}^i_{\mathcal{B}\boldsymbol{\beta}})^{-1}\mathbf{H}^i_{\mathcal{B}\mathbf{b}}$. Therefore, we seek a structured secant equation expressed in terms of the Schur complement $\mathbf{D}_i$ defined in (\ref{eq:Schur}). We suggest to choose $\hat{\mathbf{D}}^i$ so that it satisfies the structured secant equation in (\ref{eq:SchurSecant}).

\begin{equation}\label{eq:SchurSecant}
    \begin{split}
        &(\mathbf{v}_{i}')_{[\mathbf{b}]} = \left(\left[\mathbf{0}~\hat{\mathbf{D}}^i + \mathbf{H}^i_{\boldsymbol{b}\mathcal{B}}(\mathbf{H}^i_{\mathcal{B}\mathcal{B}})^{-1}\mathbf{H}^i_{\mathcal{B}\boldsymbol{b}}\right] + \left[\mathbf{H}^i_{\boldsymbol{b}\mathcal{B}}~\mathbf{0}\right]\right)\mathbf{s}_i \\
        &\hat{\mathbf{D}}^i\mathbf{s}_i' = \mathbf{v}_{i}''\ \text{with}\ \mathbf{s}_i'=(\mathbf{s}_i)_{[\mathbf{b}]}\ \text{and}\\
        &\mathbf{v}_{i}'' = (\mathbf{v}_{i}')_{[\mathbf{b}]}- \left(\left[\mathbf{0}~ \mathbf{H}^i_{\boldsymbol{b}\mathcal{B}}(\mathbf{H}^i_{\mathcal{B}\mathcal{B}})^{-1}\mathbf{H}^i_{\mathcal{B}\boldsymbol{b}}\right]+ \left[\mathbf{H}^i_{\boldsymbol{b}\mathcal{B}}~\mathbf{0}\right]\right)\mathbf{s}_i
    \end{split}
\end{equation}

Here $\mathbf{v}_{i}'$ is defined as in (\ref{eq:v_sharp}), and we use $(\mathbf{v}_{i}')_{[\mathbf{b}]}$ notation to define an indexing operation that returns elements in $\mathbf{v}_{i}'$ at the indices of $\boldsymbol{\beta}$ which contain $\mathbf{b}$. The structured equation in (\ref{eq:SchurSecant}) can again be motivated by the fact that, for a strictly quadratic function $\mathcal{L}$, the equality in (\ref{eq:SchurSecant}) is met exactly when substituting $\mathbf{D}^i$ for $\hat{\mathbf{D}}^i$. Given (\ref{eq:SchurSecant}), $\hat{\mathbf{D}}^i = \mathbf{I}\gamma_i' + \mathbf{U}'^i\mathbf{B}'^i\mathbf{U}'^{i\top}$ can readily be obtained via (limited-memory) SR1 updating as defined in (\ref{eq:SR1}) using the structured update vector pairs $(\mathbf{v}_{i}'', \mathbf{s}_i')$ in place of $(\mathbf{v}_{i}', \mathbf{s}_i)$. Subsequently, (\ref{eq:PSDSR1}) can be used to obtain a PSD (or PD) approximation $\hat{\mathbf{D}}^i_+ = \mathbf{I}\gamma_i' + \mathbf{P}'^i\boldsymbol{\Sigma}'^i_+\mathbf{P}'^{i\top}$, which ensures that $\hat{\mathbf{H}}^i_{\mathbf{b}\mathbf{b}}$ and $\hat{\mathbf{H}}^i$ are PSD (or PD) as well. Piecing all of this together, we end up with the PSD (or PD) approximation defined in (\ref{eq:psdblockHsum}).

\begin{equation}\label{eq:psdblockHsum}
    \hat{\mathbf{H}}^i_+ =
    \begin{bmatrix}
    \mathbf{H}^i_{\mathcal{B}\mathcal{B}} & \mathbf{0}\\
    \mathbf{0} & \mathbf{I}\gamma_i' \\
    \end{bmatrix} +
    \begin{bmatrix}
    \mathbf{0} & \mathbf{0}\\
    \mathbf{0} & \mathbf{H}^i_{\mathbf{b}\mathcal{B}}(\mathbf{H}^i_{\mathcal{B}\mathcal{B}})^{-1}\mathbf{H}^i_{\mathcal{B}\mathbf{b}} \\
    \end{bmatrix} +
    \begin{bmatrix}
    \mathbf{0} & \mathbf{H}^i_{\mathcal{B}\mathbf{b}}\\
    \mathbf{H}^i_{\mathbf{b}\mathcal{B}} & \mathbf{0} \\
    \end{bmatrix} +
    \begin{bmatrix}
    \mathbf{0} & \mathbf{0}\\
    \mathbf{0} & \mathbf{P}'^i\boldsymbol{\Sigma}'^i_+\mathbf{P}'^{i\top} \\
    \end{bmatrix}
\end{equation}

Here we again assume that $\mathbf{H}^i_{+\,\mathcal{B}\mathcal{B}}$ is substituted for $\mathbf{H}^i_{\mathcal{B}\mathcal{B}}$ in case the latter is not PD. Note, that each of the final three matrices in (\ref{eq:psdblockHsum}) can readily be defined as a matrix product of the form $\mathbf{P}\boldsymbol{\Sigma}_+\mathbf{K}$ required by the identity in (\ref{eq:WBinv}). Thus, one option to compute $(\hat{\mathbf{H}}_+ + \mathbf{S}_\lambda)^{-1}$, with $\hat{\mathbf{H}}_+$ now defined as in (\ref{eq:psdblockHsum}), is to apply the identity recursively. This is illustrated in (\ref{eq:WBinv1}-\ref{eq:WBinv3}), where iteration-related super-scripts (and sub-script for vectors) have again been omitted to avoid clutter.

\begin{equation}\label{eq:WBinv1}
    \left(\begin{bmatrix}
    \mathbf{H}_{\mathcal{B}\mathcal{B}} & \mathbf{0}\\
    \mathbf{0} & \mathbf{I}\gamma' +\mathbf{H}_{\mathbf{b}\mathcal{B}}(\mathbf{H}_{\mathcal{B}\mathcal{B}})^{-1}\mathbf{H}_{\mathcal{B}\mathbf{b}} \\
    \end{bmatrix} +
    \mathbf{S}_\lambda\right)^{-1} = \left(\begin{bmatrix}
    \mathbf{H}_{\mathcal{B}\mathcal{B}} & \mathbf{0}\\
    \mathbf{0} & \mathbf{I}\gamma' \\
    \end{bmatrix} + \mathbf{S}_\lambda\right)^{-1} -
    \mathbf{M}_1\mathbf{A}_1^{-1}\mathbf{N}_1\\
\end{equation}

\begin{equation}\label{eq:WBinv2}
\begin{split}
    \left(\begin{bmatrix}
    \mathbf{H}_{\mathcal{B}\mathcal{B}} & \mathbf{H}_{\mathcal{B}\mathbf{b}}\\
    \mathbf{H}_{\mathbf{b}\mathcal{B}} & \mathbf{I}\gamma' +\mathbf{H}_{\mathbf{b}\mathcal{B}}(\mathbf{H}_{\mathcal{B}\mathcal{B}})^{-1}\mathbf{H}_{\mathcal{B}\mathbf{b}} \\
    \end{bmatrix} +
    \mathbf{S}_\lambda\right)^{-1} =&\ \left(\begin{bmatrix}
    \mathbf{H}_{\mathcal{B}\mathcal{B}} & \mathbf{0}\\
    \mathbf{0} & \mathbf{I}\gamma' \\
    \end{bmatrix} + \mathbf{S}_\lambda\right)^{-1} -\\
    &\ \mathbf{M}_1\mathbf{A}_1^{-1}\mathbf{N}_1 - \mathbf{M}_2\mathbf{A}_2^{-1}\mathbf{N}_2
\end{split}
\end{equation}

\begin{equation}\label{eq:WBinv3}
    \left(\hat{\mathbf{H}}_+ + \mathbf{S}_\lambda \right)^{-1} =
    \left(\begin{bmatrix}
    \mathbf{H}_{\mathcal{B}\mathcal{B}} & \mathbf{0}\\
    \mathbf{0} & \mathbf{I}\gamma' \\
    \end{bmatrix} + \mathbf{S}_\lambda\right)^{-1} -
    \mathbf{M}_1\mathbf{A}_1^{-1}\mathbf{N}_1 -
    \mathbf{M}_2\mathbf{A}_2^{-1}\mathbf{N}_2 -
    \mathbf{M}_3\mathbf{A}_3^{-1}\mathbf{N}_3
\end{equation}

Since $\mathbf{A}_1$, $\mathbf{A}_2$, and $\mathbf{A}_3$ are of dimensions $N_\mathcal{B} * N_\mathcal{B}$, $2N_\mathcal{B} * 2N_\mathcal{B}$, and $N_u * N_u$ respectively, the implicit representation in (\ref{eq:WBinv3}) can still be more efficient than an explicit representation of $\hat{\mathbf{H}}_+$, as long as $N_\mathcal{B} << N_p$ and care is taken to ensure that the inverse from the previous step is not evaluated explicitly. Similarly, use of (\ref{eq:lbfgsTrace}) to compute the trace required in (\ref{eq:efs}) will also only be more efficient than direct evaluation when $N_\mathcal{B} << N_p$.

\subsubsection{Indefinite and Alternative PSD Approximations of the Hessian}\label{sec:idpsdapprox}

So far we have assumed that $\mathbf{H}_{\mathcal{B}\mathcal{B}}$ is PD, in order for the inverse $(\mathbf{H}_{\mathcal{B}\mathcal{B}})^{-1}$ and more generally the Schur complement $\mathbf{D}$ to be defined. In principle, (\ref{eq:psdblockHsum}) can also be used to obtain a potentially indefinite approximation $\hat{\mathbf{H}}$ of $\mathbf{H}$ directly, by replacing the inverse with a generalized inverse $(\mathbf{H}_{\mathcal{B}\mathcal{B}})^{-}$ and then working in terms of the generalized Schur complement $\mathbf{D} \coloneq \mathbf{H}/\mathbf{H}_{\mathbf{b}\mathbf{b}}= \mathbf{H}_{\mathbf{b}\mathbf{b}} - \mathbf{H}_{\boldsymbol{b}\mathcal{B}}(\mathbf{H}_{\mathcal{B}\mathcal{B}})^{-}\mathbf{H}_{\mathcal{B}\boldsymbol{b}}$. The latter can again be approximated via SR1 updating, relying on the definitions in (\ref{eq:SchurSecant}) with $(\mathbf{H}_{\mathcal{B}\mathcal{B}})^{-}$ substituted for $(\mathbf{H}_{\mathcal{B}\mathcal{B}})^{-1}$. The resulting (potentially indefinite) approximation is defined in (\ref{eq:indefblockHsum}), with $\hat{\mathbf{D}}^i=\mathbf{I}\gamma_i' + \mathbf{U}'^i\mathbf{B}'^i\mathbf{U}'^{i\top}$ now denoting the SR1 approximation to the aforementioned generalized Schur complement.

\begin{equation}\label{eq:indefblockHsum}
    \hat{\mathbf{H}}^i =
    \begin{bmatrix}
    \mathbf{H}^i_{\mathcal{B}\mathcal{B}} & \mathbf{0}\\
    \mathbf{0} & \mathbf{I}\gamma_i' \\
    \end{bmatrix} +
    \begin{bmatrix}
    \mathbf{0} & \mathbf{0}\\
    \mathbf{0} & \mathbf{H}^i_{\mathbf{b}\mathcal{B}}(\mathbf{H}^i_{\mathcal{B}\mathcal{B}})^{-}\mathbf{H}^i_{\mathcal{B}\mathbf{b}} \\
    \end{bmatrix} +
    \begin{bmatrix}
    \mathbf{0} & \mathbf{H}^i_{\mathcal{B}\mathbf{b}}\\
    \mathbf{H}^i_{\mathbf{b}\mathcal{B}} & \mathbf{0} \\
    \end{bmatrix} +
    \begin{bmatrix}
    \mathbf{0} & \mathbf{0}\\
    \mathbf{0} & \mathbf{U}'^i\mathbf{B}'^i\mathbf{U}'^{i\top} \\
    \end{bmatrix}
\end{equation}

Working in terms of the generalized Schur complement also provides a route to a PSD approximation of $\hat{\mathbf{H}}_+$ for which it is sufficient to assume that $\mathbf{H}_{\mathcal{B}\mathcal{B}}$ in (\ref{eq:indefblockHsum}) is PSD and that $(\mathbf{I} - \mathbf{H}_{\mathcal{B}\mathcal{B}}(\mathbf{H}_{\mathcal{B}\mathcal{B}})^{-})\mathbf{H}_{\mathcal{B}\boldsymbol{b}}=\mathbf{0}\Leftrightarrow \mathbf{H}_{\mathcal{B}\mathcal{B}}(\mathbf{H}_{\mathcal{B}\mathcal{B}})^{-}\mathbf{H}_{\mathcal{B}\boldsymbol{b}}=\mathbf{H}_{\mathcal{B}\boldsymbol{b}}$. In that case $\hat{\mathbf{H}}_+$ will be PSD if the generalized Schur complement $\mathbf{D}$ is PSD as well. Since the latter can be approximated via SR1 updating, a PSD approximation can readily be obtained via (\ref{eq:PSDSR1}) and substituted into (\ref{eq:indefblockHsum}), ensuring that this last requirement can be met in practice. Similarly, if a PSD approximation to $\mathbf{H}_{\mathcal{B}\mathcal{B}}$ is desired, but $(\mathbf{I} - \mathbf{H}_{\mathcal{B}\mathcal{B}}(\mathbf{H}_{\mathcal{B}\mathcal{B}})^{-})\mathbf{H}_{\mathcal{B}\boldsymbol{b}} \neq \mathbf{0}$, the off-diagonal blocks in the third matrix in (\ref{eq:indefblockHsum}) can be replaced with the projection $\mathbf{H}_{\mathcal{B}\mathcal{B}}(\mathbf{H}_{\mathcal{B}\mathcal{B}})^{-}\mathbf{H}_{\mathcal{B}\boldsymbol{b}}$. Note, that for the update in (\ref{eq:efs}) as well as the inverse computations outlined in (\ref{eq:WBinv1}-\ref{eq:WBinv3}) we still require $\hat{\mathbf{H}}_+ + \mathbf{S}_\lambda$ to be PD and the initial inverse on the right-hand side of (\ref{eq:WBinv3}) to be defined.

\subsection{Improving the Efficiency of the qEFS Update}

One might object that the structured quasi-Newton updates outlined in the previous sections are unnecessarily expensive, considering that $(\hat{\mathbf{H}}_+ + \mathbf{S}_\lambda)^{-1}$ is only really required at convergence as a substitute in (\ref{eq:efs}). Indeed, as illustrated in more detail in Appendix \ref{app:combqn} of the supplementary materials, we can simply rely on unstructured (limited-memory) BFGS updating to estimate $\hat{\boldsymbol{\beta}}$. We can then either retain the last $N_u$ pairs of update vectors $(\mathbf{v}_i,\mathbf{s}_i)$ of the quasi-Newton routine or sample $N_u$ steps (or rather perturbations) $\mathbf{s}_i$ and then set $\hat{\boldsymbol{\beta}}_{i-1} = \hat{\boldsymbol{\beta}} - \mathbf{s}_i$, where $\hat{\boldsymbol{\beta}}$ again denotes the estimate at convergence of the quasi-Newton routine \citep[cf.][]{berahas_quasi-newton_2022}. In both cases, the $\mathbf{v}_i$ are then transformed as shown in (\ref{eq:v_sharp}) to obtain $N_u$ structured update vector pairs $\mathbf{v}_i',\mathbf{s}_i$ which in turn can be used to construct $\hat{\mathbf{H}}_+$. Note, that the sampling strategy is particularly attractive when relying on (\ref{eq:psdblockHsum}), since the same $\mathbf{H}_{\mathcal{B}\mathcal{B}}$ and $\mathbf{H}_{\mathcal{B}\mathbf{b}}$ can then be used to obtain $\mathbf{v}_i'',\ i=1,...,N_u$ as defined in (\ref{eq:SchurSecant}).

\section{Performance of the qEFS Update}\label{sec:Performance}

The performance of the quasi-Newton updates described in sections (\ref{sec:structuredQN}-\ref{sec:structuredQN2}) depends on the choice for $N_u$ and, in case (\ref{eq:psdblockHsum}) is used, on the choice for $N_\mathcal{B}$. Naturally, $\mathbf{H}_{\mathcal{B}\mathcal{B}}$ in (\ref{eq:psdblockHsum}) will eventually reflect the entire negative Hessian $\mathbf{H}$ as $N_\mathcal{B} \rightarrow N_p$. In that case, use of $\hat{\mathbf{H}}_+$ as part of a qEFS update will coincide with performing the EFS update in (\ref{eq:efs}). In practice, the $N_\mathbf{b} >> 0$ case, where $N_\mathbf{b} = N_p-N_\mathcal{B}$, is more interesting -- in particular the case of $N_\mathbf{b} = N_p$ that arises when relying on (\ref{eq:PSDSR1}). 

Fortunately, the updates described here inherit many beneficial properties of conventional SR1 updating. For example, as illustrated by Theorem \ref{theorem:thessian}, a straightforward generalization of Theorem 6.1 given by \citet{nocedal_numerical_2006}, it can be shown that $\hat{\mathbf{H}}$, defined as in (\ref{eq:indefblockHsum}), equals $\mathbf{H}$ under certain conditions for strictly quadratic log-likelihood functions and any choice of $N_\mathbf{b}$ so that $N_p > N_u = N_\mathbf{b} \geq 1$.

\begin{theorem}\label{theorem:thessian}
    Let $\mathcal{L}$ be a strongly concave quadratic function of the form $\mathcal{L}(\boldsymbol{\beta}) = \mathbf{a}^\top\boldsymbol{\beta} + \frac{1}{2}\boldsymbol{\beta}^T\mathbf{H}\boldsymbol{\beta}$ with $\mathbf{H}$ symmetric. We assume that we have a set of $\mathbf{s}_i,\ i=1,...,N_u$ steps (or perturbations) with $\hat{\boldsymbol{\beta}}_{i} = \hat{\boldsymbol{\beta}}_{i-1} + \mathbf{s}_i$ and that the corresponding $N_u$ transformed $\mathbf{s}_i',\ i=1,...,N_u$, defined in (\ref{eq:SchurSecant}), are linearly independent.  $\hat{\mathbf{D}}^i = \mathbf{I}\gamma_i' + \mathbf{U}'^i\mathbf{B}'^i\mathbf{U}'^{i\top}$ again denotes the limited-memory approximation to the true generalized Schur complement $\mathbf{D} \coloneq \mathbf{H}/\mathbf{H}_{\mathbf{b}\mathbf{b}} = \mathbf{H}_{\mathbf{b}\mathbf{b}} - \mathbf{H}_{\boldsymbol{b}\mathcal{B}}(\mathbf{H}_{\mathcal{B}\mathcal{B}})^{-}\mathbf{H}_{\mathcal{B}\boldsymbol{b}}$ as defined in section \ref{sec:idpsdapprox}. Similarly, we define $\tilde{\mathbf{D}}^i = \mathbf{I}\gamma_{N_u}' + \mathbf{U}'^i\mathbf{B}'^i\mathbf{U}'^{i\top}$.
    
    Then, for any $N_\mathbf{b}$ so that $N_p > N_u = N_\mathbf{b} \geq 1$, $\hat{\mathbf{D}}^{N_u} = \mathbf{D}$ and $\hat{\mathbf{H}}^{N_u} = \mathbf{H}$, with $\hat{\mathbf{H}}^{N_u}$ defined as in (\ref{eq:indefblockHsum}), iff $(\tilde{\mathbf{D}}^{i-1}\mathbf{s}_i' - \mathbf{v}_{i}'')^\top\mathbf{s}_i' \neq 0\ \forall\ i$ with $\mathbf{v}_{i}''$ defined as in (\ref{eq:SchurSecant}).
\end{theorem}

\begin{proof}
     First, note that the $\tilde{\mathbf{D}}^i = \mathbf{I}\gamma_{N_u}' + \mathbf{U}'^i\mathbf{B}'^i\mathbf{U}'^{i\top},\ i = 1,...,N_u$ form the explicit sequence of $N_u$ SR1 updates, defined as in (\ref{eq:SR1U}), that when applied to the initial approximation $\mathbf{I}\gamma_{N_u}'$ results in $\tilde{\mathbf{D}}^{N_u} = \hat{\mathbf{D}}^{N_u}$. By assuming that $(\tilde{\mathbf{D}}^{i-1}\mathbf{s}_i' - \mathbf{v}_{i}'')^\top\mathbf{s}_i' \neq 0\ \forall\ i$, each of these updates is defined. Also, by definition of $\mathcal{L}$ and those given in (\ref{eq:SchurSecant}), $\tilde{\mathbf{D}}^{i}\mathbf{s}_{i}'  = \mathbf{v}_{i}'' = \mathbf{D}\mathbf{s}_{i}'\ \forall\ i$.
     
    Next we define the residual matrix $\mathbf{E}^i = \tilde{\mathbf{D}}^i - \mathbf{D}$ reflecting the residual between the approximation at iteration $i$ and the true generalized Schur complement \citep[cf.,][]{dennis_numerical_1996,ye_towards_2023}. Note, that because $\mathbf{v}_{i}'' = \mathbf{D}\mathbf{s}_{i}'$ we have

     $$\mathbf{v}_{i}'' - \tilde{\mathbf{D}}^{i-1}\mathbf{s}_i'= \mathbf{D}\mathbf{s}_i' - \tilde{\mathbf{D}}^{i-1}\mathbf{s}_i'=-\mathbf{E}^{i-1}\mathbf{s}'$$

     Thus, via (\ref{eq:SR1}) we can define the residual matrix recursively as
     
     \begin{equation}\label{eq:residual_matrix}
         \mathbf{E}^i = \mathbf{E}^{i-1} - \frac{\mathbf{E}^{i-1}\mathbf{s}_i'\mathbf{s}_i'^\top\mathbf{E}^{i-1}}{\mathbf{s}_i'^\top\mathbf{E}^{i-1}\mathbf{s}_i'}
     \end{equation}
    
     which when multiplied by $\mathbf{s}_{i}'$ yields
    
     $$
     \mathbf{E}^i\mathbf{s}_{i}' = \mathbf{E}^{i-1}\mathbf{s}_{i}' - \mathbf{E}^{i-1}\mathbf{s}_{i}' = 0
     $$
     
     Similarly, $\mathbf{E}^i\mathbf{a}=\mathbf{0}$ for any vector $\mathbf{a}$ in the Kernel $Ker(\mathbf{E}^{i-1})$ of $\mathbf{E}^{i-1}$. It follows, that $Ker(\mathbf{E}^{i-1}) \in Ker(\mathbf{E}^{i})$ and that the latter now includes $\mathbf{s}_{i}'$ \citep[cf.,][]{ye_towards_2023}. By the assumed linear independence of the $\mathbf{s}_i'$  the Kernel of $\mathbf{E}^i$ expands for every $i$ with $Ker(\mathbf{E}^{N_u}) = \mathbb{R}^{N_u}$, which implies that $\mathbf{E}^{N_u} = \mathbf{0}$. This in turn implies $\tilde{\mathbf{D}}^{N_u}=\hat{\mathbf{D}}^{N_u}=\mathbf{D}$. Substituting $\hat{\mathbf{D}}^{N_u}$ into (\ref{eq:indefblockHsum}) we get $\hat{\mathbf{H}}^{N_u} = \mathbf{H}$.
\end{proof}

\begin{remark}
    The proof essentially applies Lemma 3 by \citet{ye_towards_2023}, which outlined the Kernel-expanding property of the SR1 residual matrix defined here in \citep[\ref{eq:residual_matrix}; see also][]{dennis_numerical_1996}, to the approximation of the generalized Schur complement $\tilde{\mathbf{D}}$. \citet{nocedal_numerical_2006} also provide a different proof of the same result that applies when using a sequential strategy to obtain the $\mathbf{s}_i,\ i=1,...,N_u$. In fact, Theorem 6.1 given by \citet{nocedal_numerical_2006} applies directly for the $N_p = N_u = N_\mathbf{b}$ case, when (\ref{eq:SR1}) is used to obtain $\hat{\mathbf{H}}$.
    
    In contrast, theorem \ref{theorem:thessian} applies independent of the strategy chosen to obtain the $\mathbf{s}_i,\ i=1,...,N_u$ (i.e., $\mathbf{s}_i',\ i=1,...,N_u$) and thus also applies to the sampling strategy outlined in Appendix \ref{app:combqn} of the supplementary materials, since $\hat{\boldsymbol{\beta}}_{i}$ can be replaced with the maximizer of the penalized likelihood $\hat{\boldsymbol{\beta}}$ for all $i$. Theorem \ref{theorem:thessian} also applies when working in terms of the standard Schur complement since $\hat{\mathbf{H}}=\hat{\mathbf{H}}_+$ with $\hat{\mathbf{H}}_+$ defined as in (\ref{eq:psdblockH}) if we further assume that $\mathbf{H}$ is PSD (or PD) which is however not required.
\end{remark}

\begin{figure}[h!]
\caption{Simulation Study MSE Performance}
\begin{center}
\includegraphics[width=\linewidth]{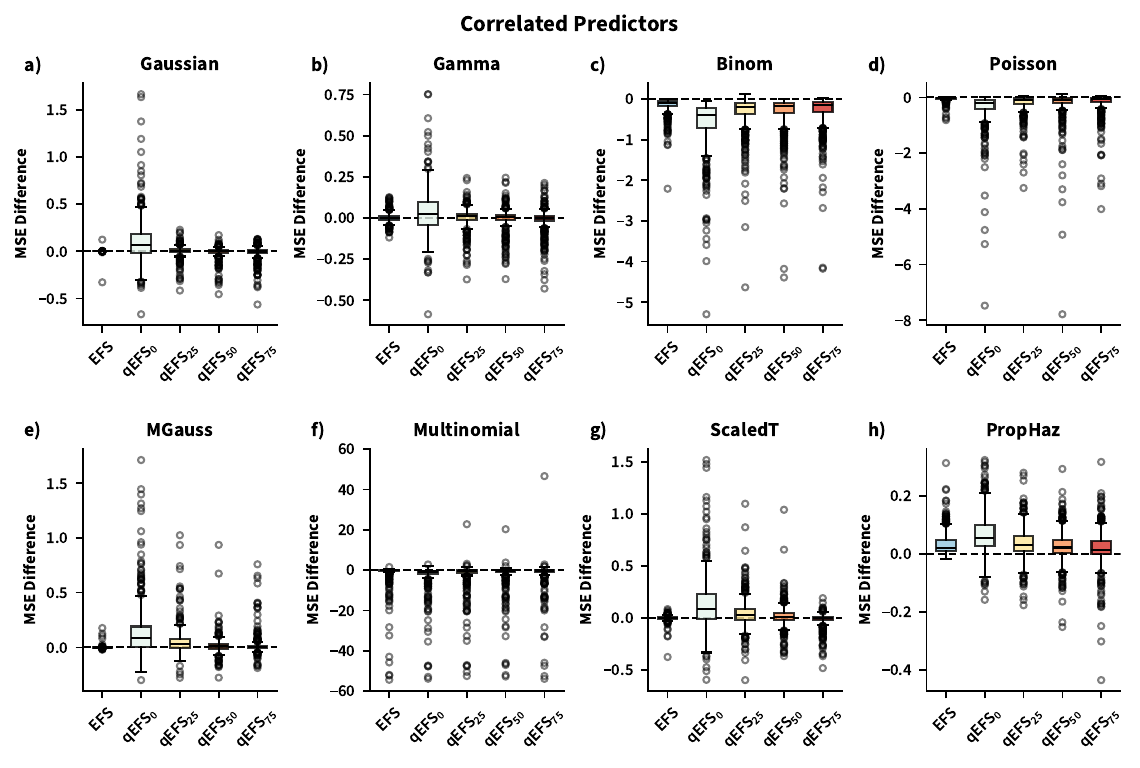}
\end{center}

        {\small
         Standardized MSE score differences from the second simulation study defined in Appendix \ref{app:sim} of the supplementary materials in which predictor variables were strongly correlated. MSE scores were first standardized by the median MSE score, computed separately per panel (i.e., likelihood). The standardized MSE scores of the method by \citet{wood_smoothing_2016} were then subtracted from the standardized MSE scores of the EFS and qEFS updates. Thus, negative values imply a better (i.e., lower) MSE score for the latter two in a particular simulation. Axes were truncated at -60 for the Multinomial models (panel f). As described in Appendix \ref{app:sim} of the supplementary materials, subscripts for the different qEFS updates indicate the percentage of coefficients included in $\mathcal{B}$.
        }

\label{fig:sim_mse}
\end{figure}

Theorem \ref{theorem:thessian}, just like Theorem 6.1 given by \citet{nocedal_numerical_2006}, applies only to strictly quadratic log-likelihood functions. Just like the regular SR1 update, the structured updates described in sections (\ref{sec:structuredQN}-\ref{sec:structuredQN2}) can however be expected to continue to produce reasonable Hessian approximations for sufficiently smooth general likelihood functions \citep[see][]{nocedal_numerical_2006}. In fact, for sufficiently smooth general likelihood functions Theorem \ref{theorem:thessian} continues to apply asymptotically, provided that the maximum step norm $max(||\mathbf{s}_i||:\ \ i=1,...,N_u) \rightarrow 0$. This is illustrated by Theorem \ref{theorem:thessian2} and the remarks given in Appendix \ref{app:theorem2} of the supplementary materials.

In practice it will often be desirable to set $N_u << N_\mathbf{b}$ and $N_\mathcal{B} << N_p$, to benefit from the efficiency of the implicit representations defined in (\ref{eq:PSDSR1}) and (\ref{eq:psdblockHsum}). The qEFS update can still be expected to produce reasonable Hessian approximations in these cases -- especially when partial approximations are used (i.e., $N_\mathcal{B} > 0$). $\mathbf{H}_{\mathcal{B}\mathcal{B}}$ and $\mathbf{H}_{\mathcal{B}\boldsymbol{b}}$ then provide additional structure for the sub-block approximation, effectively constraining the solution to the structured secant equation in (\ref{eq:SchurSecant}). As a result the approximation to $\mathbf{H}$ overall becomes more accurate \citep[e.g.,][]{dennis_convergence_1989}. To illustrate that the qEFS update continues to perform well in practice, a series of simulation studies was conducted in which the MSE performance of the EFS and four qEFS updates, all differing in the size of $\mathcal{B}$ and with $N_u=0.25N_\mathbf{b}$, was compared against the method by \citet{wood_smoothing_2016}. Appendix \ref{app:sim} of the supplementary materials contains an overview of the different simulation studies, while Figure \ref{fig:sim_mse} provides an exemplary overview of the results. Generally, MSE scores decreased and became less variable with increases to $N_\mathcal{B}$. However, already for $N_\mathcal{B}=0$ the qEFS update produced estimates of sufficient quality for practical application ($\text{qEFS}_0$ in Figure \ref{fig:sim_mse}).

\subsection{Secondary Task Performance}

Beyond estimation of $\boldsymbol{\lambda}$ we should be able to expect that, whenever $\hat{\mathbf{H}}$ is a reasonable approximation to $\mathbf{H}$ (or $\mathbb{E}\mathbf{H}$), the approximation should also perform reasonably on secondary tasks involving the Hessian. For example, approximate confidence intervals (CIs) for individual functions $f$ or entire linear predictors $\eta$ can be based on the normal approximation to the posterior $\boldsymbol{\beta}|\mathbf{y},\boldsymbol{\lambda} \sim N(\hat{\boldsymbol{\beta}},\mathcal{H}^{-1})$ \citep[e.g.,][]{marra_coverage_2012}. Similarly, selecting between two smooth models can be based on a smoothing-parameter uncertainty-corrected version of the conditional Akaike Information Criterion \citep[cAIC;][]{akaike_information_1992,greven_behaviour_2010} equal to $-2\mathcal{L}(\hat{\boldsymbol{\beta}}) + 2tr(\mathbf{V}_{\beta}\mathbf{H})$ \citep[][]{wood_smoothing_2016}. Here, $\mathbf{V}_{\beta}$ is the covariance matrix of a normal approximation to the marginal posterior $\boldsymbol{\beta}|\mathbf{y}$ defined implicitly through $\boldsymbol{\beta}|\mathbf{y},\boldsymbol{\lambda} \sim N(\hat{\boldsymbol{\beta}},\mathcal{H}^{-1})$ and the asymptotic result that $\boldsymbol{\rho} \sim N(\hat{\boldsymbol{\rho}}^*,\mathbf{H}_\rho^{-})$, where $\mathbf{H}_\rho = -\frac{\partial^2 \mathcal{V}}{\partial \boldsymbol{\rho} \partial \boldsymbol{\rho}^\top}\Big\rvert_{\hat{\boldsymbol{\rho}}^*}$,  $\boldsymbol{\rho}=log(\boldsymbol{\lambda})$, and $\hat{\boldsymbol{\rho}}^*$ denotes the maximizer\footnote{Technically, the EFS update is not guaranteed to provide an estimate $\hat{\boldsymbol{\lambda}}=exp(\hat{\boldsymbol{\rho}}^*)$ because it neglects dependence of $\mathbf{H}$ on $\boldsymbol{\lambda}$, but we can replace $\mathbf{H}_\rho$, which is also needed to compute $\mathbf{V}_\beta$, with the equivalent obtained for the final ``Taylor-approximate'' version of $\mathcal{V}$ discussed in section \ref{sec:Background} (see also Figure \ref{fig:taylor_reml}). The PQL/POI approach to GAM estimation again provides precedent \citep[see for example][section 3.4.3]{wood_generalized_2017-2}.} of $\mathcal{V}(exp(\boldsymbol{\rho}))$ \citep[e.g.,][]{wood_smoothing_2016}.

\begin{figure}[h!]
\caption{Simulation Study Confidence Interval Coverage}
\begin{center}
\includegraphics[width=\linewidth]{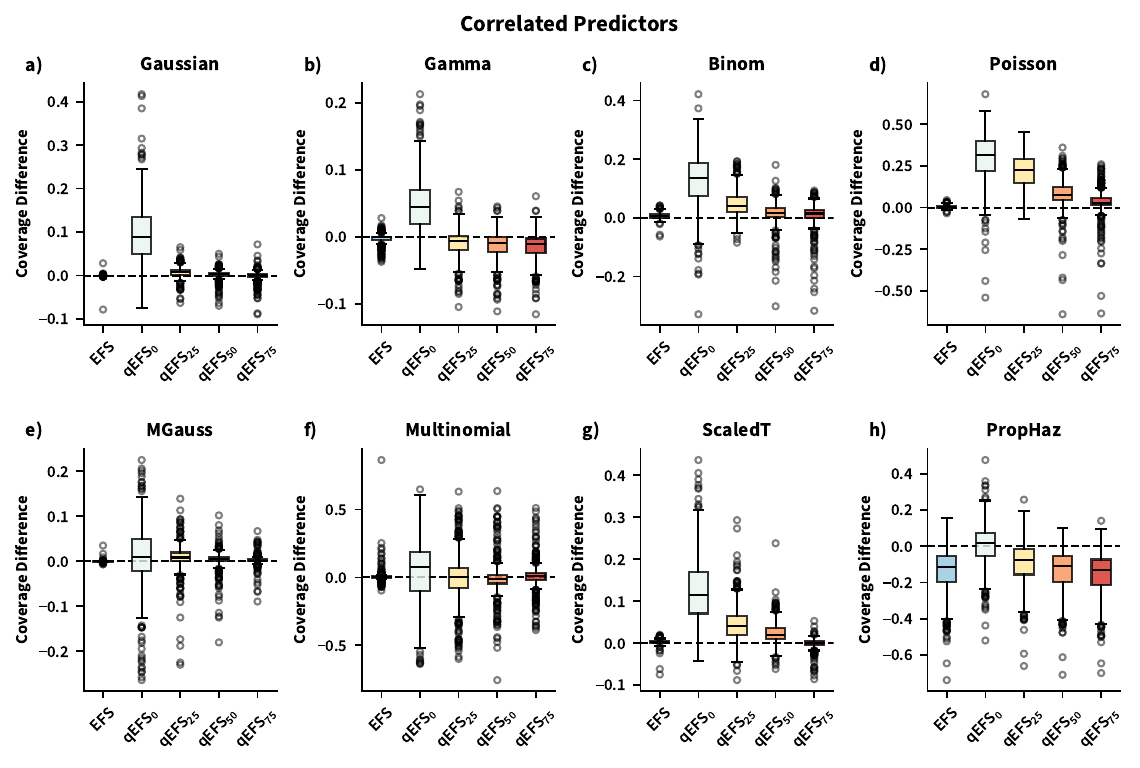}
\end{center}

        {\small
         Standardized differences in average confidence interval (CI) coverage from the second simulation study defined in Appendix \ref{app:sim} of the supplementary materials in which predictor variables were strongly correlated. Average CI coverage scores of each simulation were standardized like the MSE scores visualized in Figure \ref{fig:sim_mse}. To maintain consistency, the signs of the standardized CI coverage scores were flipped so that negative values again imply a better (i.e., higher) coverage score for the EFS and qEFS updates in a particular simulation.
        }

\label{fig:sim_coverage}
\end{figure}

To test whether $\hat{\mathbf{H}}$ continues to be a useful substitute for $\mathbf{H}$ in these secondary tasks, a second set of simulation studies was conducted (a more detailed description can again be found in Appendix \ref{app:sim} of the supplementary materials). First, we compared the average coverage of the simulated (i.e., true) linear predictor achieved by an approximate 95\% CI when relying on the EFS update, qEFS updates with varying $N_\mathcal{B}$ and $N_u=0.25N_\mathbf{b}$, and the method by \citet{wood_smoothing_2016}. Then we investigated the performance of the uncertainty-corrected cAIC based on $\hat{\mathbf{H}}$ when selecting between a smooth and random term. Choosing $\mathcal{B}$ and $N_u$ for this second task is more complicated -- Appendix \ref{app:sim} of the supplementary materials describes the strategy implemented here but we note that other strategies might work just as well (or better).

\begin{figure}[h!]
\caption{Simulation Study Model Selection}
\begin{center}
\includegraphics[width=\linewidth]{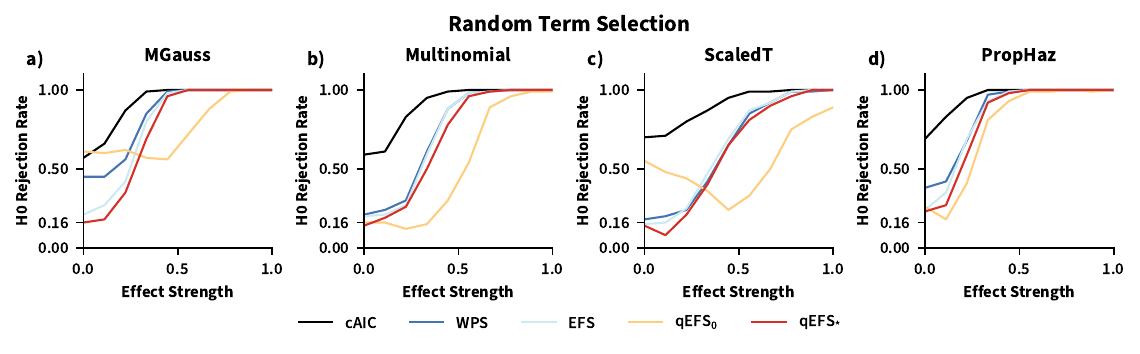}
\end{center}

        {\small
         Percentage of simulations in which the more complex model, here including 40 random intercepts, was selected over the simpler model by the cAIC for different effect strengths $e$. Nominal performance for $e=0$ would correspond to a percentage of approximately 16\%. The uncorrected cAIC (black line) performs significantly worse than that \citep[see also][]{greven_behaviour_2010,saefken_unifying_2014,wood_smoothing_2016}. cAICs based on the method by \citet{wood_smoothing_2016} (line for WPS) and the EFS update (line for EFS) are much more conservative. The cAIC based on the qEFS update for which $\mathcal{B}$ was chosen as described in Appendix \ref{app:sim} of the supplementary materials performed similarly (line for $\text{qEFS}_*$). In contrast performance of the cAIC based on the qEFS update for which $N_\mathcal{B}=0$ was more volatile.
        }

\label{fig:sim_selection}
\end{figure}

Figures \ref{fig:sim_coverage} and \ref{fig:sim_selection} again provide exemplary visualizations for this second set of simulation studies. Generally, as was the case for MSE performance, coverage improved with increases to $N_\mathcal{B}$. In contrast to MSE performance, some models (e.g., Scaled-T in Figure \ref{fig:sim_coverage}c) however continued to benefit more drastically from increases to the latter. A similar pattern emerged in the model selection study. Using, $\hat{\mathbf{H}}$ to approximately correct for smoothing parameter uncertainty generally produced a more conservative cAIC (see Figure \ref{fig:sim_selection}). Compared to the EFS update and the method by \citet{wood_smoothing_2016}, the performance of the cAIC based on the $N_\mathcal{B}=0$ approximation to $\mathbf{H}$ was however more sensitive to the specific likelihood. Choosing $\mathcal{B}$ (and $N_u$) as described in Appendix \ref{app:sim} of the supplementary materials drastically improved the criterion's stability (see the line for the $\text{qEFS}_*$ update in Figure \ref{fig:sim_selection}), resulting in performance almost identical to the EFS update.

In summary, these results indicate that sufficiently accurate estimates of the $f$ can be obtained with a relatively small value for $N_\mathcal{B}$ and $N_u$. In contrast, secondary tasks will benefit more drastically from increasing $N_\mathcal{B}$ when $N_u < N_\mathbf{b}$. This suggests that in practice, the final $\hat{\mathbf{H}}$ should either be re-computed after model estimation with $N_\mathcal{B}$ as large as possible, or at least with $N_u$ as close to $N_\mathbf{b}$ as possible -- to benefit from the convergence results described in appendix \ref{app:theorem2} of the supplementary materials.

\section{Examples}\label{sec:Examples}
\subsection{Massive Multi-level Additive Model}

As a first example we present a scaled-up version of the Gaussian random smooth model used in one of the simulation studies described in Appendix \ref{app:sim} of the supplementary materials with

$$\mu = f(x_0) + f(x_1) + f(x_2) + f(x_3) + f_s(x_0)$$

The $f_s$ are random non-linear functions \citep[e.g.,][]{wood_generalized_2017-2}. Instead of simulating 250 observations for 25 levels $s$, as was the case in the simulation studies, we now simulated data for 5,000 levels, resulting in a total of one million observations. Datasets including factor variables with 5,000 or more levels (e.g., subject identifiers) are not uncommon in the behavioral and cognitive sciences \citep[e.g.,][]{van_rij_analyzing_2019}. The scaled-up model includes 5,004 smooth functions parameterized by 50,027 coefficients but still requires estimation of only six smoothing parameters since the random functions $f_s$ all share two $\lambda$ parameters. 

\begin{figure}[h!]
\caption{Estimates for the Massive Multi-level Additive Model}
\begin{center}
\includegraphics[width=\linewidth]{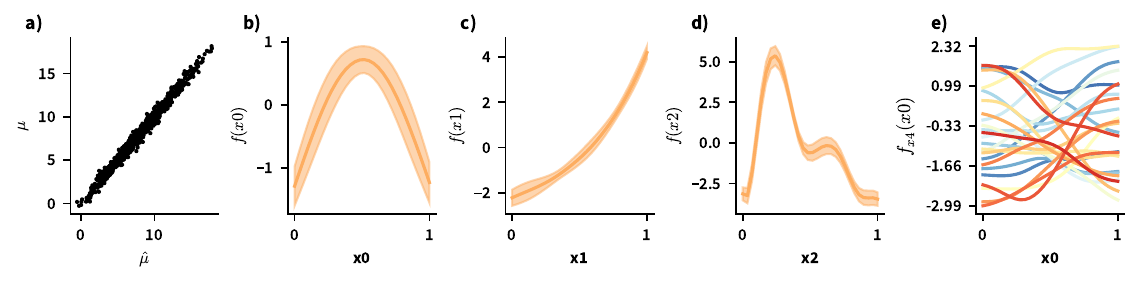}
\end{center}

        {\small
         Overview of the estimates obtained for the simulated large multi-level model. The first panel shows a scatter plot of a random sample of the true $\boldsymbol{\mu}$ against the estimate $\hat{\boldsymbol{\mu}}$. Panels b-d show estimates of the group-level smooth functions, and approximate 95\% confidence intervals, which should be compared to the functions shown in Figure \ref{fig:sim_func} of the supplementary materials. Panel e shows estimates of the random non-linear functions for a random sample of levels $s$.
        }

\label{fig:mixed_model}
\end{figure}

A similar example was presented by \citet{wood_generalized_2017}, who pointed out that estimating such a model only became feasible with the EFS update, which can exploit the sparsity in $\mathcal{H}$. The qEFS update presented here can be even more efficient, depending on the choice for $N_\mathcal{B}$ and $N_u$. In fact, as illustrated in Figure \ref{fig:mixed_model}, the model can be estimated accurately, even with $N_\mathcal{B}=0$ and $N_u=30$. This implies that $N_u=30$ distinct Eigenvalues in $\boldsymbol{\Sigma}$ (see \ref{eq:PSDSR1}) are sufficient to obtain a reasonable ``low-rank'' approximation to $\mathbf{H}$ for this model. Such an extreme approximation remains feasible, since each $s$ makes, on average, a similar (often diagonally dominant) contribution to $\mathbf{H}$. Fortunately, as also suggested by the third simulation study reported in Appendix \ref{app:sim} of the supplementary materials, this will often be the case for multi-level models. Therefore, in practice when $\mathcal{B}$ is chosen to contain all the coefficients we would like to think of as ``fixed'' (i.e., parametric terms, basis coefficients of (group-level) smooth functions) we can expect that update (\ref{eq:psdblockHsum}) will often produce estimates indistinguishable from the EFS estimates, even if $N_u << N_\mathbf{b}$ and $N_\mathbf{b} > 1e^4$ because a low-rank approximation to $\mathbf{H}_{\mathbf{b}\mathbf{b}}$ will remain feasible.

\subsection{A HsMM of Energy Prices}

As a second example we present a model of the dataset by \citet{sanchez-espigares_mswm_2021}, which is a record of the energy price in Spain between 2002 and 2008 during work days, as well as a selection of possibly predictive covariates. A Hidden Markov Model (HMM) of the data was presented previously by \citet{michelot_hmmtmb_2025}. They proposed to rely on a two-state Markov chain for the latent process model and suggested the following model for energy prices $Y_t$ while in latent state $S_t=j$:

$$
Y_t|S_t=j \sim N(\mu_{j\,t},\sigma_{j\,t})
$$

with $\mu_{j\,t} = f_j(\text{EurDol}_t)$ and $log(\sigma_{j\,t})$ modeled as a third-degree polynomial of $\text{EurDol}_t$, which denotes the exchange rate between Euro and US. Dollar \citep[e.g.,][]{michelot_hmmtmb_2025}. Here we opted to parameterize $log(\sigma_{j\,t})$ with an additional smooth function of $\text{EurDol}_t$ instead, just like the mean. Additionally, we assumed a two-state semi-Markov chain for the latent process, which allows for a fixed state transition matrix but requires explicit models of each state's duration \citep[e.g.,][]{yu_hidden_2010} -- we opted for gamma distributions normalized to a maximum state duration of 100 time-points.

\begin{figure}[h!]
\caption{Estimates for the HsMM of Energy Prices}
\begin{center}
\includegraphics[width=\linewidth]{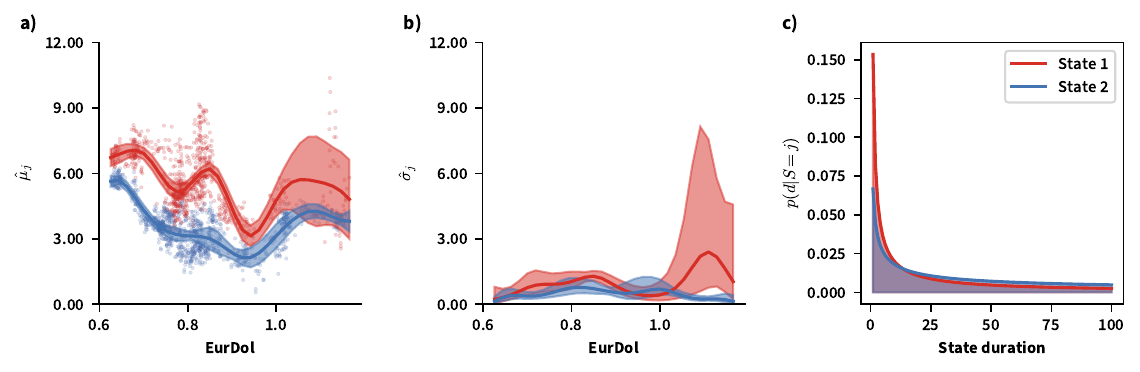}
\end{center}

        {\small
         Overview of the estimates obtained for the HsMm of Energy prices. The first panel shows the relationship between the mean energy price and exchange rate, separately for both latent states. The second panel provides the same information for the standard deviation. Shaded areas correspond to approximate 95\% confidence intervals. The scattered dots in the first panel reflect observed energy prices, which have been colored to reflect the most likely underlying latent state. The last panel shows the pdfs of the estimated Gamma distributions of each latent state's duration.
        }

\label{fig:energy_model}
\end{figure}

Deriving the necessary expressions to obtain $\mathbf{H}$ for Hidden (semi) Markov Models (HsMMs) is quite tedious \citep[e.g.,][]{lystig_exact_2002,krause_hierarchical_2026}. Additionally, computing $\mathbf{H}$ quickly becomes computationally expensive, both in terms of CPU time and memory required. These problems can be avoided when relying on the qEFS update presented here, which does not require evaluation of $\mathbf{H}$. Note, that the log-likelihood $\mathcal{L}$ of H(s)MMs is typically a complex surface featuring multiple local maxima, even for simple models like the one considered here \citep[see][for a more detailed discussion]{michelot_hmmtmb_2025}. We assume that a distinct maximizer $\hat{\boldsymbol{\beta}}$ still exists for any given $\boldsymbol{\lambda}$ (i.e., $\mathcal{L}_\lambda$ has a global maximum) and search for it via a Basin-hopping algorithm, which pairs local quasi-Newton updates with a stochastic global search strategy \citep[e.g.,][]{wales_global_1997}.

Figure \ref{fig:energy_model} shows the estimates obtained when fitting the model via the qEFS update with $N_\mathcal{B}=0$ and $N_u=30$. The estimates obtained from the ``full'' EFS update mainly distinguish themselves through narrower approximate confidence intervals for the smooth function in the model of $\sigma_{j\,t}$. Both models suggest distinct non-linear relationships between the mean and variability of energy prices and exchange rate. Additionally, the duration of the first latent state, associated with higher energy prices on average, is typically shorter than the duration of the second latent state (see Figure \ref{fig:energy_model}c).

\subsection{A Location, Scale, and Shape Tweedie Model of Mackerel Egg Densities}

For the final example, we replicate the Tweedie model of Mackerel egg counts sampled at different locations off the coast of western Europe, which was presented originally by \citet{wood_generalized_2017}. The authors suggest to model Egg $Count_i$ with a Tweedie distribution as follows:

\begin{equation}\label{eq:twmodel}
    \begin{split}
        &Count_i \sim tweedie(\mu_i,p_i,\phi_i)\ \text{with}\\
        &log(\mu_i) = log(\text{vol}_i) + f_{\mu}(\text{lo}_i,\text{la}_i) + f_{\mu}(\text{T20}_i) + f_{\mu}(\text{S20}_i) + f_{\mu}(\sqrt{\text{b.depth}}) + b_{s(i)}\text{,}\\
        &\theta_i=h(p_i) = f_{\theta}(\sqrt{\text{b.depth}})\text{, and } log(\phi_i) = f_{\phi}(\sqrt{\text{b.depth}})
    \end{split}
\end{equation}

$\text{vol}$, $\text{lo}$, $\text{la}$, $\text{T20}$, $\text{S20}$, and $\text{b.depth}$ in (\ref{eq:twmodel}) reflect the volume of water sampled, longitude, latitude, water temperature at 20 m, salinity level at 20 m, and seabed depth for a specific egg count sample. $s$ identifies the specific ship that recorded observation $i$. The model includes random intercepts $b_{s(i)} \sim N(0,\sigma_b)$ for the latter and smooth functions of all covariates but $log(\text{vol})$ which simply acts like an offset \citep[e.g.,][]{wood_smoothing_2016}. Finally, $h$ is a specific link function chosen by \citet{wood_generalized_2017} to guarantee that $1 < p_i < 2\ \forall\ i$.

The specific Tweedie model considered by \citet{wood_generalized_2017} assumes that $var(Count_i) = \phi_i\mu_i^{p_i}$, with $p_i$ and $\phi_i$ being allowed to vary as functions of $\sqrt{\text{b.depth}}$ as shown in (\ref{eq:twmodel}). The authors used the EFS update to fit this model, since the method by \citet{wood_smoothing_2016} would require third and fourth order derivatives of the corresponding log-likelihood, which are not known. Stable implementation of the partial second (mixed) derivatives of the Tweedie density with respect to the three parameters is still challenging however and can again be avoided when relying on the qEFS update presented here to estimate the model in (\ref{eq:twmodel}).

\begin{figure}[h!]
\caption{Estimates for the Tweedie Model of Mackerel Data}
\begin{center}
\includegraphics[width=\linewidth]{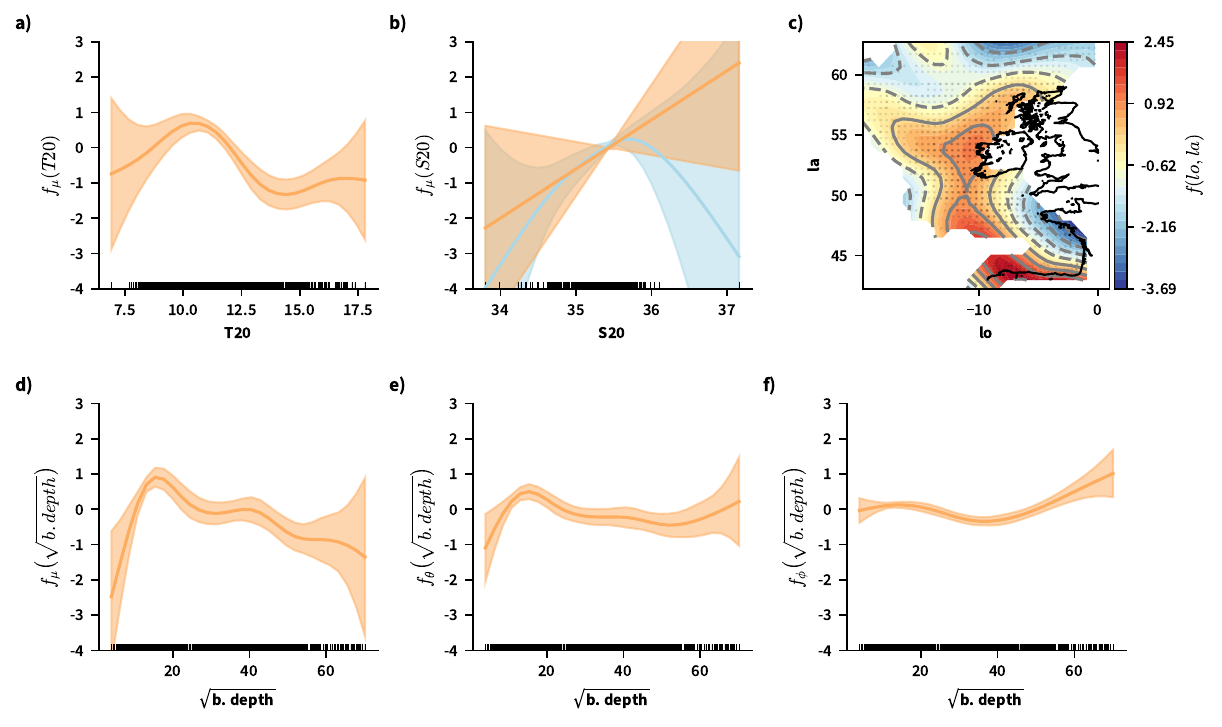}
\end{center}

        {\small
         Overview of the estimates obtained for the location, scale, and shape Tweedie model of Mackerel egg densities. Panels a, b, c, and d include function estimates and approximate 95\% confidence intervals for the smooth functions involved in the model of $\mu$. Panel b includes estimates obtained from both the qEFS (orange) and EFS (blue) update. Panels e and f depict estimates of the smooth functions involved in the models of $\theta=h(p)$ and $\phi$ respectively. Rug plots have been included to highlight covariate distributions -- note the presence of a single observation with a salinity level above 37 in panel b.
        }

\label{fig:mackerel_model}
\end{figure}

Figure \ref{fig:mackerel_model} shows the estimates obtained when fitting the model, which has $N_p=204$ coefficients with $N_\mathcal{B}=N_u=50$. The EFS and qEFS updates produce different estimates of $f_{\mu}(\text{S20})$. The corresponding EFS estimate of the latter, depicted in blue in Figure \ref{fig:mackerel_model}b, is quite non-linear, ultimately due to a single low-count observation with a comparatively high salinity level. The non-linear EFS estimate is recovered by the qEFS update when including all coefficients parameterizing $f_3$ in the set $\mathcal{B}$\footnote{In practice, this might be prudent in general for functions of covariates with few extreme values.} -- which was not the case originally since the content of $\mathcal{B}$ was selected at random. Doing so does not noticeably affect any other function estimates and partial residual plots (not shown here) are anyway reasonable for both estimates. For the remaining smooth functions, both the EFS and qEFS update produce very similar estimates. Notably, both models suggest a distinct bump in magnitude of all functions of seabed depth around $\sqrt{\text{b.depth}}\approx14.14$, coinciding with the continental shelf break \citep[e.g.,][]{wood_generalized_2017}.

\section{Discussion}\label{sec:Discussion}
Here we presented the partial structured quasi-EFS (qEFS) update, principally a first-order method to estimate the smoothing parameters of general smooth models which nonetheless allows to incorporate available information about the true Hessian. In contrast to existing methods sharing the same goal, notably the EFS update by \citet{wood_generalized_2017} and the Newton method by \citet{wood_smoothing_2016}, the qEFS update thus theoretically only requires an implementation of the model's log-likelihood to enable estimation, since the corresponding gradient can in praxis be obtained via automatic differentiation (AD) or replaced with a finite difference approximation. For more exotic smooth models, like the Hidden (semi) Markov and Tweedie models of section \ref{sec:Examples}, this drastically reduces the theoretical work necessary upfront.

The qEFS update also offers more precise control about the computing resources required during estimation, since the number $N_u$ of update vectors used to form the secant approximation as well as the number of columns constrained to match the exact Hessian $N_\mathcal{B}$ can readily be adjusted. In theory there is a trade-off between accuracy of the secant approximations and efficiency: as $N_\mathcal{B}$ increases, the method eventually and naturally converges to the (second-order) EFS update. With $N_\mathcal{B}\leq N_p$, asymptotic convergence of the secant approximations is only guaranteed for sufficiently smooth general likelihood functions and with $N_u=N_\mathbf{b}$ (see Appendix \ref{app:theorem2} of the supplementary materials). Convergence to the true Hessian in practice with $N_\mathcal{B} < N_p$ can only be guaranteed for quadratic $\mathcal{L}$ -- and again requires $N_u=N_\mathbf{b}$. However, the examples and simulation studies presented here illustrate that in practice the qEFS update continues to perform well with $N_u \leq N_\mathbf{b}$ and $N_\mathcal{B} < N_p$, likely because the $N_\mathcal{B}$ exact columns provide additional structure which enables a high-quality ``low-rank'' approximation to $\mathbf{H}_{\mathbf{b}\mathbf{b}}$.

Additionally, the EFS update by \citet{wood_generalized_2017} also optimizes $\mathcal{V}$ exactly only in the quadratic case -- because only for a quadratic $\mathcal{L}$ will it generally be the true that $\mathbf{H}$ does not depend on $\boldsymbol{\lambda}$. This has not been detrimental to the success of the method, likely because in practice the EFS update not only produces high-quality estimates of $\boldsymbol{\lambda}$ for general likelihood functions but also often is more robust to numerical problems plaguing theoretically more complex methods like the one by \citet{wood_smoothing_2016} (see the discussion in Appendix \ref{app:sim} of the supplementary materials). The examples and simulation studies presented here again show that this remains the case when relying on the qEFS update. 

Notably, the structured secant methods underlying the qEFS update will also be useful beyond the single scenario considered here -- essentially another empirical Bayes approach to smoothing parameter estimation \citep[cf.][]{wood_smoothing_2016}. More generally, $(\hat{\mathbf{H}}_+ + \mathbf{S}_\lambda)^{-1}$ or an approximation to

$$\begin{bmatrix}
\mathbf{V}_\beta & \mathbf{0}\\
\mathbf{0} & \mathbf{H}_\rho^-\\
\end{bmatrix}$$

based on $\hat{\mathbf{H}}_+$, with $\mathbf{V}_\beta$ and $\mathbf{H}_\rho^-$ defined as in section \ref{sec:Performance}, could readily be used as the (initial) inverse metric of HMC samplers of $\boldsymbol{\beta}|\mathbf{y},\boldsymbol{\lambda}$ and $\boldsymbol{\beta},\boldsymbol{\lambda}|\mathbf{y}$ (given a prior for $\boldsymbol{\lambda}$) respectively \citep[or alternatively as proposal distributions for Metropolis samplers; ][]{hoffman_no-u-turn_2014,betancourt_conceptual_2018,wood_inference_2020}. Initialized like this, samplers will potentially converge after only a very brief adaptation phase, which could drastically reduce the CPU time required to sample the respective posteriors. The adaptation phase could also be used to further refine the inverse metric if necessary, which further alleviates the risk that $\hat{\mathbf{H}}$ might be a poor approximation to $\mathbf{H}$. Notably, this initialization strategy comes at no additional theoretical cost, since HMC samplers anyway require the ability to evaluate $\mathcal{L}$ and corresponding gradient \citep[e.g.,][]{hoffman_no-u-turn_2014}. Similarly, future work could explore the feasibility of substituting $\hat{\mathbf{H}}$ for $\mathbf{H}$ as part of the (simplified) nested integrated Laplace approximation \citep[INLA; see][]{rue_approximate_2009,wood_simplified_2019}, which could again enable (approximate) inference about $\boldsymbol{\beta}$ and $\boldsymbol{\lambda}$ at the low theoretical cost of a first order method.

\section*{Implementation in Software}
The different variants of the qEFS update have been implemented in the open-source \texttt{mssm} Python toolbox \citep[v $\geq 1.2.5$][]{krause_mixed-sparse-smooth-model_2025}, available on GitHub at \url{https://github.com/JoKra1/mssm} and distributed via the Python Packaging Index.

\section*{Supplementary materials}
We have attached Supplementary materials with additional information to the main manuscript. Additionally, we set up a supplementary GitHub repository containing code to replicate all simulation studies and examples as well as the visualizations of the results. The repository is available at: \url{https://github.com/JoKra1/lSQEFS_supplementary}.

\section*{Acknowledgements}
We want to thank Jeremias Krause for his feedback on an earlier version of this manuscript which substantially improved the paper.

\printbibliography

\newpage
\supplement
\pagestyle{plain}
\begin{refsection}
\section*{Supplementary Materials for ``Structured Secant Methods to Select Smoothing Parameters For General Smooth Models''}

Joshua Krause\textsuperscript{1}, Jelmer P. Borst\textsuperscript{1}, and Jacolien van Rij\textsuperscript{1}

\textsuperscript{1} Bernoulli Institute for Mathematics, Computer Science, and Artificial Intelligence, Faculty of Science and Engineering, University of Groningen, Groningen, Netherlands

\subsection*{Content}
\begin{itemize}
\item Appendix A: The Extended Fellner \& Schall Update by Wood \&
Fasiolo
\item Appendix B: The Implicit Eigendecomposition Approach
\item Appendix C: Combining Structured \& Unstructured Quasi-Newton Updates for Efficient Smooth Model Estimation
\item Appendix D: Convergence to the Hessian in the General Case
\item Appendix E: Simulation Studies
\end{itemize}

\section{The Extended Fellner \& Schall Update by Wood \& Fasiolo}\label{app:efs}

This Appendix provides a brief review of the EFS update by \citet{wood_generalized_2017}. The update, designed to approximately optimize the Laplace-approximate Bayesian marginal log-likelihood $\mathcal{V}(\boldsymbol{\lambda})$ defined in (\ref{EQ:reml}) of the main manuscript, can be motivated by consideration of the first partial derivatives of $\mathcal{V}(\boldsymbol{\lambda})$, defined in (\ref{eq:reml_deriv}). Note, that since $\frac{\partial \mathcal{L}_\lambda}{\partial \boldsymbol{\beta}}\Big\rvert_{\hat{\boldsymbol{\beta}}}=\mathbf{0}$, computation of $\frac{\partial \mathcal{L}_\lambda}{\partial \lambda_r} + \frac{\partial \mathcal{L}_\lambda}{\partial \boldsymbol{\beta}}\frac{d \boldsymbol{\beta}}{d \lambda_r}$ given $\hat{\boldsymbol{\beta}}$ simplifies drastically \citep[e.g.,][]{wood_generalized_2017-1}.

\begin{equation}\label{eq:reml_deriv}
\frac{\partial \mathcal{V}}{\partial \lambda_r}=-\frac{\hat{\boldsymbol{\beta}}^\top\mathbf{S}^r\hat{\boldsymbol{\beta}}}{2} + \frac{1}{2}\frac{\partial log|\mathbf{\mathbf{S}_\lambda}|_+}{\partial \lambda_r} - \frac{1}{2}\frac{\partial log|\mathcal{H}|}{\partial \lambda_r}.
\end{equation}

In general, $\frac{\partial log|\mathcal{H}|}{\partial \lambda_r}=tr\left(\mathcal{H}^{-1}\left[\frac{\partial \mathbf{H}}{\partial \lambda_r} + \mathbf{S}^r\right]\right)$ with the inverse being replaced by a pseudo-inverse in case of a generalized determinant like $log|\mathbf{\mathbf{S}_\lambda}|_+$ \citep[see for example Appendix B of][]{wood_fast_2011}. However, assuming instead that $\mathbf{H}$ does not depend on $\boldsymbol{\lambda}$, the derivative defined in (\ref{eq:reml_deriv}) is non-zero only as long as $\left(tr\left(\mathbf{S}^{-}_\lambda \mathbf{S}^r\right) - tr\left(\mathcal{H}^{-1} \mathbf{S}^r\right) \right) \neq \boldsymbol{\hat{\beta}}^\top\mathbf{S}^r\boldsymbol{\hat{\beta}}$ with the sign of the derivative being determined by the direction of the inequality \citep[][]{wood_generalized_2017}.

Exploiting this relationship, the authors propose the EFS update defined in (\ref{eq:efs}) of the main manuscript. Note, that the update is only exact if $\mathbf{H}$ really does not depend on $\boldsymbol{\lambda}$, since otherwise $tr\left(\mathcal{H}^{-1} \mathbf{S}^r\right) \neq tr\left(\mathcal{H}^{-1}\left[\frac{\partial \mathbf{H}}{\partial \lambda_r} + \mathbf{S}^r\right]\right)$. Additionally, note that application of the update in practice requires limits on the minimum and maximum values for any $\lambda_r$, to address the case of the numerator or denominator approaching zero \citep[e.g.,][]{wood_generalized_2017}.

\section{The Implicit Eigendecomposition Approach}\label{app:ImplicitEigen}

This Appendix provides a review of the implicit Eigendecomposition approach proposed originally by \citet{burdakov_efficiently_2017} used to compute the positive semi-definite matrix $\hat{\mathbf{H}}_+^{i}$ that is closest, in terms of the Frobenious norm, to $\hat{\mathbf{H}}^{i}$ without explicitly forming an Eigendecomposition of the latter. The following relies on the definitions given in (\ref{eq:SR1}) and (\ref{eq:PSDSR1}) of the main manuscript. Specifically, let $\mathbf{B}^i$ and $\mathbf{U}^i$ be defined as in (\ref{eq:SR1}) and $\mathbf{Q}^i\mathbf{R}^i=\mathbf{U}^i$ again denote a thin QR decomposition. Finally, let $\mathrm{U}^i\boldsymbol{\Sigma}^i\mathrm{U}^{i\top}$ again denote the implicit Eigendecomposition defined in (\ref{eq:PSDSR1}), so that $\mathbf{R}^i\mathbf{B}^i\mathbf{R}^{i\top} = \mathrm{U}^i\boldsymbol{\Sigma}^i\mathrm{U}^{i\top}$. Super-scripts indicating the iteration $i$ of the quasi-Newton routine are omitted from now on to avoid clutter.

As mentioned, the approach considered here was originally proposed by \citet{burdakov_efficiently_2017}. However, \citet{erway_efficiently_2015} provide an alternative derivation that better highlights the relation between the implicit decomposition of $\hat{\mathbf{H}}^i$ and the explicit Eigendecomposition. We thus follow their approach. They start by considering the full QR-decomposition

\begin{equation}\label{eq:qr}
\mathrm{Q}\mathrm{R}=\left[\mathbf{Q~~\tilde{\mathbf{Q}}}\right]  \begin{bmatrix}
\mathbf{R}\\
\mathbf{0}\\
\end{bmatrix} = \mathbf{U}
\end{equation}

Here, $\mathrm{Q}$ is a $N_p*N_p$ orthogonal matrix of which $\mathbf{Q}$ makes up the first $N_u$ columns and $\mathbf{R}$ is a $N_u*N_u$ upper-triangular matrix \citep[][]{erway_efficiently_2015}. Next, \citet{erway_efficiently_2015} define the Eigendecomposition

\begin{equation}\label{eq:qr2}
\begin{split}
\mathrm{R}\mathbf{B}\mathrm{R}^\top &=
\begin{bmatrix}
\mathbf{R}\mathbf{B}\mathbf{R}^\top & \mathbf{0}\\
\mathbf{0} & \mathbf{0}\\
\end{bmatrix}=\begin{bmatrix}
\mathrm{U}\boldsymbol{\Sigma}\mathrm{U}^\top & \mathbf{0}\\
\mathbf{0} & \mathbf{0}\\
\end{bmatrix}
\\&=\mathrm{U''}\boldsymbol{\Sigma}''\mathrm{U''}^\top~\text{with}\\&~\mathrm{U''}=\begin{bmatrix}
\mathrm{U} & \mathbf{0}\\
\mathbf{0} & \mathbf{I}\\
\end{bmatrix}~\text{and}~
\boldsymbol{\Sigma}''=\begin{bmatrix}
\boldsymbol{\Sigma} & \mathbf{0}\\
\mathbf{0} & \mathbf{0}\\
\end{bmatrix}
\end{split}
\end{equation}

Then, $\mathrm{P}\boldsymbol{\Sigma}''\mathrm{P}^\top$, with $\mathrm{P} = \mathrm{Q}\mathrm{U''}$, is the explicit Eigendecomposition of $\mathbf{U}\mathbf{B}\mathbf{U}^\top$ and $\hat{\mathbf{H}} = \mathrm{P}(\mathbf{I} \gamma + \boldsymbol{\Sigma}'')\mathrm{P}^\top$ is the explicit Eigendecomposition of $\hat{\mathbf{H}}$ \citep[since $\mathrm{Q}$ is orthogonal;][]{erway_efficiently_2015}. In combination with (\ref{eq:qr2}), this provides the insight already stated in the main manuscript: negative Eigenvalues of $\hat{\mathbf{H}}$ originate from negative Eigenvalues in $\boldsymbol{\Sigma}$ with a magnitude greater than $\gamma$. This makes it straightforward to compute the nearest positive semi-definite matrix $\hat{\mathbf{H}}_+ = \mathrm{P}(\mathbf{I}\gamma + \Sigma''_+)\mathrm{P}^\top$, where

\begin{equation}\label{eq:evShift}
\boldsymbol{\Sigma}''_+=\begin{bmatrix}
\boldsymbol{\Sigma}_+ & \mathbf{0}\\
\mathbf{0} & \mathbf{0}\\
\end{bmatrix}~\text{with}~\boldsymbol{\Sigma}_{+\,l\,l} = \boldsymbol{\Sigma}_{l\,l} - min(0,\boldsymbol{\Sigma}_{l\,l} + \gamma),
\end{equation}

so that all negative Eigenvalues are shifted just enough that they become zero when adding $\gamma$ \citep[see][]{higham_computing_1988}. In a final step, \citet{erway_efficiently_2015} show that in practice $\hat{\mathbf{H}}_+$ can be computed without forming $\mathrm{Q}$, $\mathrm{U''}$, and $\boldsymbol{\Sigma}''$. They point out, that $\gamma$ and the decomposition $\mathbf{R}\mathbf{B}\mathbf{R}^\top$ already ``implicitly'' represent the entire spectrum of $\hat{\mathbf{H}}$. Thus, computing the thin QR-decomposition $\mathbf{Q}\mathbf{R}=\mathbf{U}$ is sufficient to obtain $\hat{\mathbf{H}}_+$ \citep[see also][]{burdakov_efficiently_2017}. Specifically, from this decomposition and (\ref{eq:qr2}), which implicitly requires only $\mathbf{R}$ but not $\mathrm{R}$, it follows that

\begin{equation}\label{eq:evEquality}
\begin{split}
\hat{\mathbf{H}} &= \mathbf{I}\gamma + \mathrm{P}\boldsymbol{\Sigma}''\mathrm{P}^\top\\&=\mathbf{I}\gamma + \mathbf{P}\boldsymbol{\Sigma}\mathbf{P}^\top,
\end{split}
\end{equation}

where $\mathbf{P}=\mathbf{Q}\mathrm{U}$ is defined as in (\ref{eq:PSDSR1}) of the main manuscript. Substituting $\boldsymbol{\Sigma}_+$ for $\boldsymbol{\Sigma}$ in (\ref{eq:evEquality}) then finally provides

\begin{equation}\label{eq:SR1UPHessianPSDAPD}
\begin{split}
\hat{\mathbf{H}}_+ &= \mathbf{I}\gamma + \mathrm{P} \boldsymbol{\Sigma}''_+\mathrm{P}^\top\\&= \mathbf{I}\gamma + \mathbf{P}\boldsymbol{\Sigma}_+\mathbf{P}^\top.
\end{split}
\end{equation}

The second line in (\ref{eq:SR1UPHessianPSDAPD}) is the result stated in the main text.

\section{Combining Structured \& Unstructured Quasi-Newton Updates for Efficient Smooth Model Estimation}\label{app:combqn}

This Appendix outlines two estimation strategies that enable use of the structured secant approximations to $\mathbf{H}$ in practice, while remaining efficient. The point is that, as discussed in the main manuscript, $(\hat{\mathbf{H}}_+ + \mathbf{S}_\lambda)^{-1}$ is only really required at convergence as a substitute in (\ref{eq:efs}) and that, for the task of obtaining $\hat{\boldsymbol{\beta}}$, it will be far more efficient to rely on unstructured (limited-memory) BFGS updating of $\hat{\mathcal{H}}^{i-1}$ as described at the beginning of section (\ref{sec:Method}). Since $\mathcal{H}$ should itself be positive definite at the converged estimate -- assuming that the model is not miss-specified and that no parameters are unidentifiable -- relying on BFGS updating will also generally be appropriate for this task.

Fortunately, the two approaches can readily be combined: (Limited-memory) BFGS updating can be used to estimate $\hat{\boldsymbol{\beta}}$ in a first step and only the last $N_u$ pairs of update vectors $(\mathbf{v}_i,\mathbf{s}_i)$ are retained. The $\mathbf{v}_i$ are then transformed as shown in (\ref{eq:v_sharp}) to obtain $N_u$ structured update vector pairs $\mathbf{v}_i',\mathbf{s}_i$. Finally, the latter can be used to construct $\hat{\mathbf{H}}_+$ as shown in (\ref{eq:PSDSR1}) and subsequently to compute the pq-EFS update by replacing the trace in (\ref{eq:efs}) with the approximation in (\ref{eq:lbfgsTrace}).

This sequential strategy can still become expensive when relying on (\ref{eq:psdblockHsum}) as discussed in section \ref{sec:structuredQN2}, in principle requiring $N_u$ evaluations of $\mathbf{H}_{\mathcal{B}\mathcal{B}}$ and $\mathbf{H}_{\mathcal{B}\mathbf{b}}$. This can be avoided by sampling $N_u$ steps (or rather perturbations) $\mathbf{s}_i$ and then setting $\hat{\boldsymbol{\beta}}_{i-1} = \hat{\boldsymbol{\beta}} - \mathbf{s}_i$, where $\hat{\boldsymbol{\beta}}$ again denotes the estimate at convergence of the quasi-Newton routine (i.e., the maximizer of $\mathcal{L}_\lambda$ given $\boldsymbol{\lambda}$). Having all $N_u$ steps end at $\hat{\boldsymbol{\beta}}$ implies $\hat{\boldsymbol{\beta}}_i = \hat{\boldsymbol{\beta}},\ i=1,...,N_u$ so that the same $\mathbf{H}_{\mathcal{B}\mathcal{B}}$ and $\mathbf{H}_{\mathcal{B}\mathbf{b}}$ can be used to obtain $\mathbf{v}_i'',\ i=1,...,N_u$ as defined in (\ref{eq:SchurSecant}).

More generally, having all $N_u$ steps end at $\hat{\boldsymbol{\beta}}$ has also been shown to result in more accurate Hessian approximations compared to the sequential approach outlined above \citep[][]{berahas_quasi-newton_2022}. While \citet{berahas_quasi-newton_2022} suggest to sample the perturbations uniformly, we can also exploit the fact that the normal approximation to the posterior $\boldsymbol{\beta}|\mathbf{y},\boldsymbol{\lambda} \sim N(\hat{\boldsymbol{\beta}},\mathcal{H}^{-1})$ has good theoretical support \citep[see][]{wood_smoothing_2016}. This suggests to instead draw the perturbations from $N(\mathbf{0},\omega(\hat{\mathbf{H}}_+ + \mathbf{S}_\lambda)^{-1})$, where $\hat{\mathbf{H}}_+$ is formed as for the sequential strategy (without incorporating $\mathbf{H}_{\mathcal{B}\mathcal{B}}$ and $\mathbf{H}_{\mathcal{B}\mathbf{b}}$), or $N(\mathbf{0},\omega(\mathbf{I}\gamma + \mathbf{S}_\lambda)^{-1})$ if the first alternative is too expensive to evaluate. $\omega$ here denotes a scaling factor that can for example be adjusted in a $\boldsymbol{\lambda}$-specific warm-up phase via dual-averaging to achieve a pre-specified Metropolis acceptance probability for the sampled perturbations \citep[e.g.,][]{hoffman_no-u-turn_2014}.

\section{Convergence to the Hessian in the General Case}\label{app:theorem2}

This Appendix discusses convergence of the structured Hessian updates described in sections \ref{sec:structuredQN}-\ref{sec:structuredQN2} of the main manuscript for more general likelihood functions, which are not necessarily (strictly) concave or quadratic. Theorem \ref{theorem:thessian2}, again a generalization of Theorem 6.2 given by \citet{nocedal_numerical_2006}, investigates this for the update in section \ref{sec:structuredQN2}. The subsequent remark considers the update in section \ref{sec:structuredQN} as well.

\begin{theorem}\label{theorem:thessian2}
    Let $\mathcal{L}_\lambda(\boldsymbol{\beta}) = \mathcal{L}(\boldsymbol{\beta}) - \frac{1}{2}\boldsymbol{\beta}^\top\mathbf{S}_\lambda\boldsymbol{\beta}$ be the penalized log-likelihood for a twice continuously differentiable log-likelihood function $\mathcal{L}$ for any given $\boldsymbol{\lambda}$. We assume that $\hat{\boldsymbol{\beta}}$, the maximizer of $\mathcal{L}_\lambda$, exists and is unique. Further, we assume that the Hessian $\mathcal{H}$ of the negative penalized log-likelihood, with $\mathcal{H} = \mathbf{H} + \mathbf{S}_\lambda$ and $\mathbf{H}=-\frac{\partial^2 \mathcal{L}}{\partial \boldsymbol{\beta} \partial \boldsymbol{\beta}^\top}\Big\rvert_{\hat{\boldsymbol{\beta}}}$, is bounded and Lipschitz continuous in a neighborhood of $\hat{\boldsymbol{\beta}}$. Additionally, we assume $\mathbf{s}_i,\ i\in \mathbb{N}$ to be an infinite sequence of perturbations with $\hat{\boldsymbol{\beta}}_{i-1} = \hat{\boldsymbol{\beta}} - \mathbf{s}_i$. We assume that the perturbations are obtained in a way that the running set $\left[...,\mathbf{s}'_{i-1},\mathbf{s}'_i\right]$ of $min(i,N_u)$ transformed perturbations $\mathbf{s}_i' = (\mathbf{s}_i)_{[\mathbf{b}]}$, relying on the same indexing operation defined in (\ref{eq:SchurSecant}), remains linearly independent and that $\lim_{i\to\infty} \epsilon_i = 0$ where $\epsilon_i = max(||\mathbf{s}_j||:\ max(1,i-N_u+1) \leq j \leq i)$ defines the maximum perturbation norm in the running set for iteration $i$. Finally, we define $\hat{\mathbf{D}}^i$ and $\hat{\mathbf{H}}^i$ as in Theorem (\ref{theorem:thessian}), with the aforementioned running set forming the queue for the limited-memory SR1 approximation $\hat{\mathbf{D}}_i$.

    Then, Theorem (\ref{theorem:thessian}) from the main manuscript applies asymptotically, ensuring that
    
    $$\lim_{i\to\infty} ||\mathbf{H} - \hat{\mathbf{H}}^i|| = 0$$

    iff $(\tilde{\mathbf{D}}^{j-1}\mathbf{s}_i' - \mathbf{v}_{i}'')^\top\mathbf{s}_i' \neq 0\ \forall\ max(1,i-N_u+1) \leq j \leq i$, where $\tilde{\mathbf{D}}$ is defined as in Theorem (\ref{theorem:thessian}).
\end{theorem}

\begin{proof}
    In what follows, let $\mathbf{D}$, $\tilde{\mathbf{D}}$, and $\hat{\mathbf{H}}$ all be defined as in Theorem (\ref{theorem:thessian}).
    We first note that because  $\mathcal{H}$ is bounded and Lipschitz continuous in a neighborhood of $\hat{\boldsymbol{\beta}}$, so is $\mathbf{H}$. This implies that for any $\hat{\boldsymbol{\beta}}_{i}$ in a sufficiently close neighborhood of $\hat{\boldsymbol{\beta}}$ we have $||\mathbf{H}^i - \mathbf{H}|| \leq L||\hat{\boldsymbol{\beta}}_{i-1} - \hat{\boldsymbol{\beta}}|| = L||\mathbf{s}_i||$ with $\mathbf{H}^i=-\frac{\partial^2 \mathcal{L}}{\partial \boldsymbol{\beta} \partial \boldsymbol{\beta}^\top}\Big\rvert_{\hat{\boldsymbol{\beta}}_{i-1}}$.
    
    Next, we again consider the residual matrix $\mathbf{E}^i = \tilde{\mathbf{D}}^i - \mathbf{D}$, again reflecting the residual between the approximation at iteration $i$ and the true Schur complement \citep[cf.,][see also Theorem \ref{theorem:thessian}]{dennis_numerical_1996,ye_towards_2023}. Because we assume that $(\tilde{\mathbf{D}}^{j-1}\mathbf{s}_i' - \mathbf{v}_{i}'')^\top\mathbf{s}_i' \neq 0\ \forall\ max(1,i-N_u+1) \leq j \leq i$, we can again guarantee that $\tilde{\mathbf{D}}^i\mathbf{s}_i' = \mathbf{v}_i''\ \forall\ i$.
    
    The latter no longer holds exactly for a general log-likelihood when substituting the true Schur complement $\mathbf{D}$ for $\tilde{\mathbf{D}}^i$. However, by the mean value theorem and in a sufficiently close neighborhood of $\hat{\boldsymbol{\beta}}$, we have

    \begin{equation}\label{eq:MVT}
        \begin{split}
        \mathbf{v}_{j}' &= \tilde{\mathbf{H}}^j\mathbf{s}_j\ \text{with}\ \tilde{\mathbf{H}}^j =\frac{\partial^2 \mathcal{L}}{\partial \boldsymbol{\beta} \partial \boldsymbol{\beta}^\top}\Big\rvert_{\tilde{\boldsymbol{\beta}}_j}\\
        &=\mathbf{H}\mathbf{s}_j + \left(\tilde{\mathbf{H}}^j-\mathbf{H}\right)\mathbf{s}_j\\
        &=\mathbf{H}\mathbf{s}_j\ + \mathbf{r}_j\ \text{with}\ ||\mathbf{r}_j|| \leq L\epsilon_i||\mathbf{s}_j^2||
        \end{split}
    \end{equation}

    for all $max(1,i-N_u+1) \leq j \leq i$ (i.e., all elements in the running set at iteration $i$) where $\tilde{\boldsymbol{\beta}}_j$ lies on the line from $\hat{\boldsymbol{\beta}}_{j-1}$ to $\hat{\boldsymbol{\beta}}$ \citep[cf.,][]{conn_convergence_1991}. Thus, we also have

    \begin{equation}\label{eq:approxSchurSecant}
        \begin{split}
        &(\mathbf{v}_{j}')_{[\mathbf{b}]} = \left(\left[\mathbf{0}~\mathbf{D} + \mathbf{H}_{\boldsymbol{b}\mathcal{B}}(\mathbf{H}_{\mathcal{B}\mathcal{B}})^{-}\mathbf{H}_{\mathcal{B}\boldsymbol{b}}\right] + \left[\mathbf{H}_{\boldsymbol{b}\mathcal{B}}~\mathbf{0}\right]\right)\mathbf{s}_j + \mathbf{r}'_j\text{,}\\
        &\text{that is }\mathbf{D}\mathbf{s}_j' + \mathbf{r}'_j = \mathbf{v}_{j}''\ \text{with}\ \mathbf{s}_j' = (\mathbf{s}_j)_{[\mathbf{b}]}\text{,}\ ||\mathbf{r}'_j|| \leq ||\mathbf{r}_j||\text{, and}\\ &\mathbf{v}_{j}'' = (\mathbf{v}_{j}')_{[\mathbf{b}]}- \left(\left[\mathbf{0}~ \mathbf{H}_{\boldsymbol{b}\mathcal{B}}(\mathbf{H}_{\mathcal{B}\mathcal{B}})^{-}\mathbf{H}_{\mathcal{B}\boldsymbol{b}}\right]+ \left[\mathbf{H}_{\boldsymbol{b}\mathcal{B}}~\mathbf{0}\right]\right)\mathbf{s}_j
        \end{split}
    \end{equation}

    Finally, based on the definitions in (\ref{eq:MVT}-\ref{eq:approxSchurSecant}) we have that
    
    $$\mathbf{v}_{j}'' - \tilde{\mathbf{D}}^{j-1}\mathbf{s}_j' = \mathbf{D}\mathbf{s}_j'- \tilde{\mathbf{D}}^{j-1}\mathbf{s}_j' + \mathbf{r}'_j =-\mathbf{E}^{j-1}\mathbf{s}' + \mathbf{r}'_j$$
    
    which again permits recursive definition of a residual matrix \citep[cf.,][]{ye_towards_2023}, now for the running set of perturbations at iteration $i$, as illustrated in (\ref{eq:residual_matrix2}).
    
    \begin{equation}\label{eq:residual_matrix2}
    \mathbf{E}^j = \mathbf{E}^{j-1} + \frac{(-\mathbf{E}^{j-1}\mathbf{s}' + \mathbf{r}'_j)(-\mathbf{E}^{j-1}\mathbf{s}' + \mathbf{r}'_j)^\top}{(-\mathbf{E}^{j-1}\mathbf{s}' + \mathbf{r}'_j)^\top\mathbf{s}_j'}
    \end{equation}
    
    In contrast to the quadratic case outlined in Thereom \ref{theorem:thessian}, multiplication of (\ref{eq:residual_matrix2}) by $\mathbf{s}_{j}'$ only yields
    
    \begin{equation}\label{eq:kernel2}
    ||\mathbf{E}^j\mathbf{s}_{j}'|| = ||\mathbf{E}^{j-1}\mathbf{s}_{j}' - \mathbf{E}^{j-1}\mathbf{s}_{j}' + \mathbf{r}'_j|| = ||\mathbf{r}'_j|| \neq 0
    \end{equation}
    
    However, the final inequality in (\ref{eq:kernel2}) asymptotically becomes an equality for all $max(1,i-N_u+1) \leq j \leq i$, since $\lim_{i\to\infty} ||\epsilon_i|| = 0$. In fact, since $||\mathbf{r}'_j|| \leq L\epsilon_i||\mathbf{s}_j^2||$, $\frac{||\mathbf{E}^j\mathbf{s}_{j}'||}{||\mathbf{s}_{j}'||} \rightarrow 0$ for all $max(1,i-N_u+1) \leq j \leq i$, which ensures that asymptotically the residual matrix for the running set, defined in (\ref{eq:residual_matrix2}), behaves just like the residual matrix for the quadratic case outlined in Theorem \ref{theorem:thessian} of the main manuscript so that $\lim_{i\to\infty} ||\mathbf{E}^i|| = 0$ by the repeated Kernel expansion discussed in Theorem \ref{theorem:thessian}. This implies $\lim_{i\to\infty} ||\mathbf{H} - \hat{\mathbf{H}}^i|| = 0$.
\end{proof}

\begin{remark}
    The proof again relies on previous work by \citet{ye_towards_2023} and \citet{conn_convergence_1991}. The latter presented the same result for sequential strategies, like the one described in Appendix \ref{app:combqn}, where the perturbations defined by Theorem \ref{theorem:thessian2} become quasi-Newton steps. In particular, the proof by \citet{conn_convergence_1991} provides the same error bound stated here in (\ref{eq:MVT}). Notably, when working with sequences of quasi-Newton steps, the independence assumption of Theorem \ref{theorem:thessian2} needs to be replaced with the assumption that the transformed sequential steps $\mathbf{s}_i' = (\mathbf{s}_i)_{[\mathbf{b}]}$ remain uniformly linearly independent \citep[see][for corresponding Theorems and Proofs]{conn_convergence_1991,nocedal_numerical_2006,ye_towards_2023}. For example, Theorem 6.2 given by \citet{nocedal_numerical_2006} provides the same result stated in Theorem \ref{theorem:thessian2} for the $N_\mathbf{b}=N_p$ case that arises when applying the update from section \ref{sec:structuredQN} (sequentially).

    Working in terms of a more general sequence of perturbations, rather than sequential steps, has the advantage that this provides important insights into how the theoretical convergence guarantee of Theorem \ref{theorem:thessian2} can be utilized to optimize performance in practice. Specifically, consider that one option to define the sequence of perturbations required by Theorem \ref{theorem:thessian2} is to assume that each perturbation $\mathbf{s}_i$ in the current running set is sampled from $N(\mathbf{0},\omega_i\mathbf{I})$, where $\omega_i \rightarrow 0$ as $i \rightarrow \infty$. This suggests that in practice, whenever $N_u=N_{\mathbf{b}}$, we should sample a single (running) set of perturbations from $N(\mathbf{0},\omega\mathbf{I})$ with $\omega$ as small as possible. This is complicated by the fact that ``what is possible'' will depend on the specific implementation in computer code and the machine precision. Fortunately, the sampling strategy described in Appendix \ref{app:combqn} provides a straightforward way to estimate $\omega$ in practice. Also note, that when $N_u < N_{\mathbf{b}}$ the performance of the sampling strategy will generally benefit from the more complex covariance matrix options described in Appendix \ref{app:combqn}.
\end{remark}

\section{Simulation Studies}\label{app:sim}

This Appendix describes the setup of the simulation studies performed to validate the qEFS update's performance in practice and highlights the main results. Code to replicate all of the simulation studies as well as visualizations of all results are provided in the supplementary GitHub repository for this paper.

\paragraph{MSE Performance (Simulations 1-3)}

\begin{figure}[h!]
\caption{Functions Used in Simulation Studies}
\begin{center}
\includegraphics[width=\linewidth]{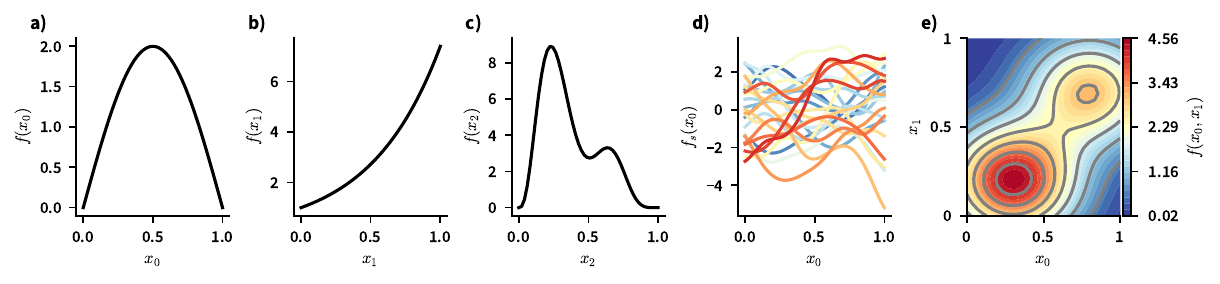}
\end{center}

        {\small
         Overview of the true smooth functions used across simulation studies. Panels a-c show smooth functions $f(x0)$, $f(x_1)$, and $f(x_2)$ used in previous simulations by \citet{gu_minimizing_1991}, $f(x_3)$ was simply set to zero everywhere. Panel d shows a random selection of the random non-linear functions $f_s(x_0)$ used in the third simulation study. Panel e shows the non-linear interaction $f(x_1,x_2)$ used in the second simulation study. The function was used previously by \citet{wood_straightforward_2013}.
        }

\label{fig:sim_func}
\end{figure}

To illustrate the performance of the qEFS update in practice, a series of simulation studies was conducted. The setup generally matched the one described by \citet{wood_smoothing_2016}. Simulations were performed for four GAMs and four more general smooth models (Cox proportional Hazard, multivariate Gaussian, Multinomial, and Scaled-T models) and involved at least one linear predictor of the form $\eta = f(x_0) + f(x_1) + f(x_2) + f(x_3)$. In case models required multiple linear predictors the functions were distributed across these. Two additional sets of simulations were performed: in one the additive effect $f(x_1) + f(x_2)$ was replaced with a non-linear interaction $f(x_1,x_2)$ and in the other random non-linear curves $f_s(x_0)$ were added to obtain a multi-level smooth model. Figure \ref{fig:sim_func} provides an overview of the functions used across simulations. Simulations were repeated for uncorrelated and correlated predictors to assess robustness.

The mean squared error (MSE) between the true linear predictor(s) $\boldsymbol{\eta}$ and the predicted version $\hat{\boldsymbol{\eta}}$ was used to compare the performance of the standard EFS update as well as four different versions of the structured qEFS update presented here, with $N_\mathcal{B} \in \lbrace0,0.25N_F,0.5N_F,0.75N_F\rbrace$ and $N_u=0.25N_\mathbf{b}$, to the method by \citet{wood_smoothing_2016}. $N_F$ here denotes the number of ``fixed'' coefficients, i.e., $N_F$ was the same for the first simulation and the multi-level one, but the ratio $N_F/N_p$ was substantially smaller in the second one. Also, for the first two simulations $N_F=N_p$. Which of the ``fixed'' coefficients to include in $\mathcal{B}$ was randomly determined. The EFS \& qEFS models were estimated using the \texttt{mssm} Python toolbox \citep[][]{krause_mixed-sparse-smooth-model_2025} and \texttt{mgcv} in R \citep[][]{wood_generalized_2017-2} was used to estimate the models based on the method by \citet{wood_smoothing_2016}. Note, that if the model is a GAM, \texttt{mgcv} falls back to the method by \citet{wood_fast_2011}, which can be considered a GAM-specific implementation of the more general framework presented by \citet{wood_smoothing_2016}.

Figure \ref{fig:sim_mse} of the main manuscript provides an exemplary overview of the results for the second simulation. Both the standard EFS and qEFS updates performed well when compared directly to the method by \citet{wood_smoothing_2016}. MSE scores were similar on average for most likelihoods, except for the Multionimial, Binomial, and Poisson cases -- here the EFS and qEFS updates consistently outperformed the method by \citet{wood_smoothing_2016}. \texttt{mgcv} also reported multiple convergence failures for the Binomial case (for correlated predictors only) and repeatedly terminated with an error for the Multinomial case (for correlated predictors only). In contrast, the EFS and qEFS updates never failed outright. Regarding the latter, MSE scores typically decreased and became less variable with increases to $N_\mathcal{B}$. However, already for $N_\mathcal{B}=0$ the qEFS update produced estimates of sufficient quality for practical application. Notably, performance of the qEFS update remained similar in the first and last simulations, differing only in the presence of the $f_s$ and thus the ratio $N_F/N_p$. This suggests that large computational savings will be possible for multi-level models.

\paragraph{Secondary Tasks (Simulations 4-5)}

To test whether $\hat{\mathbf{H}}$ continues to be a useful substitute for $\mathbf{H}$ in secondary tasks (e.g., confidence interval (CI) coverage and model selection), we first evaluated the average coverage of $\boldsymbol{\eta}$ achieved by an approximate 95\% CI around $\hat{\boldsymbol{\eta}}$ obtained as part of the simulation studies described in the previous section. Figure \ref{fig:sim_coverage} in the main manuscript again provides an exemplary illustration of the results from the second simulation study involving a non-linear interaction $f(x_1,x_2)$.

We also conducted two simulation studies to evaluate $\hat{\mathbf{H}}$ in the context of model selection via the cAIC. The linear predictor was again of the form used in the previous studies. However, in one set of simulations $f(x_0)$ was now replaced with $ef(x_0)$, while 40 random intercepts drawn from $N(0,e)$ were added to $\eta$ in another set. 10 equidistantly spaced values were considered for $e\in [0,1]$, reflecting the effect strength \citep[cf.][]{wood_smoothing_2016}. Since this results in many more simulations, only the general smooth models of the previous simulations were considered for this study.

The choice about which coefficients to include in $\mathcal{B}$ is more complicated when using $\hat{\mathbf{H}}$ for model selecting. For other simulations $\mathcal{B}$ was simply chosen at random, but when comparing models it would intuitively make more sense for both models to share the same set $\mathcal{B}$, to try and ensure that the approximation error ends up being more or less the same for both models. The problem is that even though $N_\mathbf{b}$ would be the same for both models, both models might still feature very different $\mathbf{H}_{\mathbf{b}\mathbf{b}}$ blocks at convergence -- and it is not clear that the same $N_u$ would work well to approximate them both. This could result in differences between models in the approximation error, which could influence the outcome of the comparison.

To address these problems, we chose $\mathcal{B}$ to include all coefficients involved in the selection, since these are of main interest, plus a single coefficient present in both models (the intercept of the first linear predictor, if present). We also set $N_u=0.75N_\mathbf{b}$, so that the approximation error would generally be low. For both studies we also again estimated one model with $N_\mathcal{B}=0$ so that we could compare both qEFS updates to the EFS update by \citet{wood_generalized_2017} and the method by \citet{wood_smoothing_2016}. Figure \ref{fig:sim_selection} in the main manuscript illustrates the results for the random term selection.
\printbibliography[title={References}]
\end{refsection}

\end{document}